\documentclass[11pt]{article}
\usepackage{fullpage}
\usepackage{amsmath,amsfonts,amsthm,amssymb}
\usepackage{url}
\usepackage{algorithm}
\usepackage{algorithmic}
\usepackage{bbm}
\usepackage{float}
\usepackage{framed}
\usepackage{enumerate}

\usepackage{color}
\usepackage[usenames,dvipsnames,svgnames,table]{xcolor}
\usepackage[colorlinks=true, linkcolor=red, urlcolor=blue, citecolor=gray]{hyperref}
\usepackage{array}
\usepackage[para]{footmisc}

\makeatletter
\newcommand{\thickhline}{%
    \noalign {\ifnum 0=`}\fi \hrule height 1pt
    \futurelet \reserved@a \@xhline
}
\newcolumntype{"}{@{\hskip\tabcolsep\vrule width 1pt \hskip\tabcolsep}}
\makeatother
\newcolumntype{'}{@{\hskip\tabcolsep\vrule width 1pt}}
\makeatother
\newcolumntype{`}{@{\vrule width 1pt \hskip\tabcolsep}}
\makeatother
\newcommand{\specialcell}[2][c]{%
  \begin{tabular}[#1]{@{}c@{}}#2\end{tabular}}

\let\Pr\relax
\DeclareMathOperator*{\Pr}{\mathbb{P}}
\DeclareMathOperator{\tr}{tr}
\DeclareMathOperator{\rank}{rank}
\newcommand{\R}{\mathbb{R}}

\DeclareMathOperator*{\argmin}{arg\,min}

\makeatletter
\newtheorem*{rep@theorem}{\rep@title}
\newcommand{\newreptheorem}[2]{%
\newenvironment{rep#1}[1]{%
 \def\rep@title{#2 \ref{##1}}%
 \begin{rep@theorem}}%
 {\end{rep@theorem}}}
\makeatother

\newtheorem{theorem}{Theorem}
\newreptheorem{theorem}{Theorem}

\newtheorem{corollary}[theorem]{Corollary}
\newtheorem{lemma}[theorem]{Lemma}
\newtheorem*{lemma*}{Lemma}
\newreptheorem{lemma}{Lemma}

\newtheorem{claim}[theorem]{Claim}
\newtheorem{definition}[theorem]{Definition}

\newcommand{\ol}{\overline}
\newcommand{\inprod}[1]{\left\langle #1 \right\rangle}

\newcommand{\norm}[1]{\|#1\|}
\newcommand{\abs}[1]{\lvert#1\rvert}
\newcommand{\bs}[1]{\boldsymbol{#1}}
\newcommand{\bv}[1]{\mathbf{#1}}

  \usepackage{nth}
  \usepackage{intcalc}

\title{Dimensionality Reduction for $k$-Means Clustering and Low Rank Approximation}

\author{Michael B. Cohen, Sam Elder, Cameron Musco, Christopher Musco, M\u{a}d\u{a}lina Persu
\\ \textit{Massachusetts Institute of Technology, EECS and Mathematics}
\\ \textit{Cambridge, MA 02139, USA}
\\ Email: \{micohen,same,cnmusco,cpmusco,mpersu\}@mit.edu}

\date{}

\begin{document}
\maketitle

\begin{abstract}
We show how to approximate a data matrix $\bv{A}$ with a much smaller sketch $\bv{\tilde A}$ that can be used to solve a general class of \emph{constrained k-rank approximation} problems to within $(1+\epsilon)$ error. Importantly, this class of problems includes $k$-means clustering and unconstrained low rank approximation (i.e. principal component analysis). By reducing data points to just $O(k)$ dimensions, our methods generically accelerate any exact, approximate, or heuristic algorithm for these ubiquitous problems.

For $k$-means dimensionality reduction, we provide $(1+\epsilon)$ relative error results for many common sketching techniques, including random row projection, column selection, and approximate SVD.
%
%
For approximate principal component analysis, we give a simple alternative to known algorithms that has applications
in the streaming setting. Additionally, we extend recent work on column-based matrix reconstruction, giving column subsets that not only `cover' a good subspace for $\bv{A}$, but can be used directly to compute this subspace.

Finally, for $k$-means clustering, we show how to achieve a $(9+\epsilon)$ approximation by Johnson-Lindenstrauss projecting data points to just $O(\log k/\epsilon^2)$ dimensions. This gives the first result that leverages the specific structure of $k$-means to achieve dimension independent of input size and sublinear in $k$. 


\end{abstract}

\thispagestyle{empty}
\clearpage
\setcounter{page}{1}

\section{Introduction}

Dimensionality reduction has received considerable attention in the study of fast linear algebra algorithms. The goal is to approximate a large matrix $\bv{A}$ with a much smaller $\emph{sketch}$ $\bv{\tilde A}$ such that solving a given problem on $\bv{\tilde A}$ gives a good approximation to the solution on $\bv{A}$. 
This can lead to faster runtimes, reduced memory usage, or decreased distributed communication.
Methods such as random sampling and Johnson-Lindenstrauss projection have been applied to a variety of problems including matrix multiplication, regression, and low rank approximation \cite{Halko:2011, Mahoney:2011}.

Similar tools have been used for accelerating $k$-means clustering. While exact $k$-means clustering is  NP-hard \cite{Aloise:2009:NES:1519378.1519389,Mahajan:2009:PKP:1507094.1507122}, effective heuristics 
and provably good approximation algorithms are known \cite{1056489,kanungo2002local, 1366265,Arthur2007,Har-Peled:2007:SCK:1229133.1229135}. 
Dimensionality reduction seeks to generically accelerate 
any of these algorithms by reducing the dimension of the data points being clustered. In other words, we want a sketch $\bv{\tilde A}$ with many fewer columns than the original data matrix $\bv{A}$. An approximately optimal $k$-means clustering for $\bv{\tilde A}$ should also be approximately optimal for $\bv{A}$.


\subsection{Big Picture}
We obtain a variety of new results on dimensionality reduction for both $k$-means clustering and $k$-rank approximation (also known as singular value decomposition or principal component analysis). In the later case, we will use $\bv{\tilde A}$ to find a nearly optimal $k$-dimensional basis for reconstructing the columns of $\bv{A}$ -- i.e., an approximate set of left singular vectors.

We start by noting that both problems are special cases of a general \emph{constrained  $k$-rank approximation} problem \cite{nphardkmeans}, which also includes problems related to sparse and nonnegative PCA \cite{papailiopoulos2013sparse, yuan2013truncated, asteris2014nonnegative}. Then, following the coreset definitions of \cite{feldman2013turning}, we introduce the concept of a \emph{projection-cost preserving sketch}, an approximation where the sum of squared distances of $\bv{\tilde A}$'s columns from any $k$-dimensional subspace (plus a fixed constant independent of the subspace) is multiplicatively close to that of $\bv{A}$.
This ensures that the cost of any $k$-rank projection of $\bv{A}$ is well approximated by $\bv{\tilde A}$ and thus, we can solve the general constrained $k$-rank approximation problem approximately for $\bv{A}$ using $\bv{\tilde A}$.

Next, we give several simple and efficient approaches for obtaining projection-cost preserving sketches with $(1+\epsilon)$ relative error. 
All of these techniques simply require computing an SVD, multiplying by a random projection, random sampling, or some combination of the three. These methods have well developed
implementations, are robust, and can be accelerated for sparse or otherwise structured data. As such, we do not focus heavily on specific implementations or runtime analysis. 
We do show that our proofs are amenable to approximation and acceleration in the underlying sketching techniques -- for example, it is possible to use fast approximate SVD algorithms, sparse Johnson-Lindenstrauss embeddings, and inexact sampling probabilities.

In addition to the applications in this paper, we hope that projection-cost preserving sketches will be useful in developing future randomized matrix algorithms.
They relax the guarantee of \emph{subspace embeddings}, which have received significant attention in recent years \cite{sarlos2006improved,Clarkson:2013}. Subspace embedding sketches require that $\norm{\bv{x\tilde A}} \approx \norm{\bv{x}\bv{ A}}$ simultaneously for all $\bv{x}$, which in particular implies that $\bv{\tilde A}$ preserves the cost of \emph{any} column projection of $\bv{A}$\footnote{$\norm{(\bv{I} - \bv{P})\bv{A}}_F \approx \norm{(\bv{I} - \bv{P})\bv{\tilde A}}_F$ for any  projection matrix $\bv{P}$.}. However, in general $\bv{\tilde A}$ will require at least $O(rank(\bv{A}))$ columns.
On the other hand, our projection-cost preserving sketches only work for projections with rank at most $k$, but  only require $O(k)$ columns. 

\subsection{Dimensionality Reduction Results}
In Table \ref{results_table} we summarize each of our dimensionality reduction results, showing a variety of methods for obtaining projection-cost preserving sketches. For each method, we note how many dimensions (columns) are required for a sketch $\bv{\tilde A}$ that achieves $(1+\epsilon)$ error.  We compare to prior work, most of which focuses on constructing sketches for $k$-means clustering, but applies to general constrained $k$-rank approximation as well. One exception for non-oblivious random projection is noted since no prior results were known for $k$-means or the general constrained problem.

\renewcommand*{\thefootnote}{\fnsymbol{footnote}}
\vspace{-5pt}
\begin{table}[H]
\begin{center}
\begin{tabular}{`>{\centering}m{3.2cm}">{\centering}m{1.8cm}|c|c">{\centering}m{1.65cm}|c|c'}
\thickhline
 & \multicolumn{3}{c"}{\textbf{Previous Work}} & \multicolumn{3}{c'}{\textbf{Our Results}} \\
\thickhline
\textbf{Technique}  & \textbf{Reference} & \textbf{Dimensions} & \textbf{Error} & \textbf{Theorem}  &\textbf{Dimensions} & \textbf{Error}\\
\hline
SVD & \specialcell{\cite{nphardkmeans} \\ \cite{feldman2013turning}} & \specialcell{$k$ \\ $O(k/\epsilon^2)$} & \specialcell{$2$ \\ $1+\epsilon$} & Thm \ref{exact_svd} & $\lceil k/\epsilon \rceil$ & $1+\epsilon$ \\
\hline
Approx. SVD & \cite{DBLP:journals/corr/abs-1110-2897} & $k$ & $2+\epsilon$ & Thm \ref{approx_svd},\ref{general_svd}&$\lceil k/\epsilon \rceil$ & $1+\epsilon$ \\
\hline
Random Projection & \cite{NIPS2010_3901} & $O(k/\epsilon^2)$ & $2+\epsilon$ & Thm \ref{rp_theorem}\\Thm \ref{logk_approx} & \specialcell{$O(k/\epsilon^2)$\\  $O(\log k / \epsilon^2)$} & \specialcell{$1+\epsilon$\\ $9+\epsilon$ \footnotemark[2]} \\
\hline
\specialcell{Non-oblivious\\ Randomized\\ Projection} & \cite{sarlos2006improved} & $O(k/\epsilon)$ & $1+\epsilon$ \footnotemark[3]  & Thm \ref{non_oblivious_full_theorem} 
& $O(k/\epsilon)$ & $1+\epsilon$ \\
\hline
\specialcell{Feature Selection\\ (Random Sampling)} & \cite{DBLP:conf/nips/BoutsidisMD09, DBLP:journals/corr/abs-1110-2897} & $O(k \log k / \epsilon^2)$ & $3+\epsilon$ &Thm \ref{sample}&$O(k \log k / \epsilon^2)$ & $1+\epsilon$ \\
\hline
\specialcell{Feature Selection\\ (Deterministic)} & \cite{boutsidis2013deterministic} & $r$, $k < r < n$ & $O(n/r)$ &Thm \ref{bss_theorem} 
& $O(k / \epsilon^2)$ & $1+\epsilon$ \\
\thickhline
\end{tabular}
\caption{Summary of new dimensionality reduction results.}
\label{results_table}
\end{center}
\end{table}
\footnotetext[2]{$k$-means clustering only.}
\footnotetext[3]{$k$-rank approximation only.}
\renewcommand*{\thefootnote}{\arabic{footnote}}
\vspace{-15pt}

The smallest dimension projection-cost preserving sketches can be obtained by projecting $\bv{A}$'s rows onto its top $\lceil k/\epsilon \rceil$ right singular vectors (identified using a partial singular value decomposition).
Our analysis improves on \cite{feldman2013turning}, which requires an $O(k/\epsilon^2)$ rank approximation. However, we note that our proof nearly follows from work in that paper. 

Due to the expense of computing a partial SVD, we show that any nearly optimal set of $\lceil k/\epsilon \rceil$ right singular vectors  
also suffices. This result
improves on a $(2+\epsilon)$ bound in \cite{DBLP:journals/corr/abs-1110-2897} and allows for the application of
fast approximate SVD algorithms based on Krylov subspace methods or more recent randomized techniques \cite{Halko:2011}.
SVD sketches offer some unique practical advantages. $k$ is typically small so the lack of constant factors and $1/\epsilon$ dependence (vs. $1/\epsilon^2$) can be significant. We also show that a smaller sketch suffices when $\bv{A}$'s spectrum is not uniform, a condition that is simple to check in practice.

While our SVD based dimensionality reduction results are valuable for $k$-means clustering and other constrained problems, they are not useful for the unconstrained approximate SVD problem itself -- finding $\bv{\tilde A}$ would be just as hard as solving the problem directly. Nevertheless, we give projection-cost preserving sketches based on random projection and feature selection that are useful in both the constrained and unconstrained setting. 
These results are based on a unified proof technique that relies on a reduction to a spectral approximation problem. 
The approach allows us to tighten and generalize a fruitful line of work in \cite{DBLP:conf/nips/BoutsidisMD09, NIPS2010_3901, DBLP:journals/corr/abs-1110-2897,boutsidis2013deterministic}, which were the first papers to address dimensionality reduction for $k$-means using random projection and feature selection. They inspired our general proof technique.

Specifically, we show that a $(1+\epsilon)$ error projection-cost preserving sketch can be obtained by randomly projecting  $\bv{A}$'s rows to $O(k/\epsilon^2)$ dimensions -- i.e., multiplying 
on the right by a Johnson-Lindenstrauss matrix with $O(k/\epsilon^2)$ columns. 
Sampling $O(k\log k/\epsilon^2)$ columns or using \emph{BSS selection} (a deterministic algorithm based on \cite{batson2012twice}) to choose $O(k/\epsilon^2)$ columns  also suffices. Our results improve on constant factor bounds in \cite{DBLP:conf/nips/BoutsidisMD09,NIPS2010_3901,boutsidis2013deterministic,DBLP:journals/corr/abs-1110-2897}.


Our random projection result gives the 
lowest communication relative error distributed algorithm for $k$-means, 
improving on \cite{liang2013distributed,balcan2014improved,kannan2014principal}. It also
gives an \emph{oblivious} dimension reduction technique for computing the unconstrained SVD, providing an alternative to the algorithms in \cite{sarlos2006improved,Clarkson:2013} that has applications in the streaming setting. 
We complete the picture by showing that the \emph{non-oblivous} 
technique 
in \cite{sarlos2006improved,Clarkson:2013} generalizes to constrained $k$-rank approximation. This method multiplies $\bv{A}$ on the left by a Johnson-Lindenstrauss matrix with just $O(k/\epsilon)$ rows and then projects onto the row span of this smaller matrix.

For low rank approximation, our feature selection results are similar to column-based matrix reconstruction \cite{deshpande2006matrix,guruswami2012optimal,boutsidis2014near,boutsidis2014optimal}, but we give stronger guarantees at the cost of worse $\epsilon$ dependence. We discuss the strong connection with this line of work in Section \ref{spectral_norm_proofs}.

Finally, for general constrained $k$-rank approximation, it is not possible to reduce to dimension below $\Theta(k)$. However, we conclude by showing that it \emph{is possible} to do better for $k$-means clustering by leveraging the problem's specific structure. Specifically, randomly projecting to $O(\log k /\epsilon^2)$ dimensions 
is sufficient to obtain a $(9+\epsilon)$ approximation to the optimal clustering. This gives the first $k$-means sketch 
with dimension independent of the input size and sublinear in $k$. It is simple to show via the standard Johnson-Lindenstrauss lemma that $O(\log n/\epsilon^2)$ dimension projections yield $(1+\epsilon)$ error, also specifically for $k$-means \cite{DBLP:journals/corr/abs-1110-2897}. Our results offers significantly reduced dimension and we are interested in knowing whether our $(9+\epsilon)$ error bound can be improved.

\subsection{Road Map}
\begin{description}
\item[Section \ref{prelims}]  Review notation and linear algebra basics. Introduce \emph{constrained low rank approximation} and demonstrate that $k$-means clustering is a special case of the problem.
\item[Section \ref{error_matrices}] Introduce \emph{projection-cost preserving sketches} and their applications. 
\item[Section \ref{sec:sufficient_conditions}] Overview our approach and give sufficient conditions for projection-cost preservation.
\item[Section \ref{svds}] Prove that projecting onto $\bv{A}$'s top $\lceil k/\epsilon \rceil$ singular vectors or finding an approximately optimal $\lceil k/\epsilon \rceil$-rank approximation gives a projection-cost preserving sketch.
\item[Section \ref{block_section}] Reduce projection-cost preservation to spectral norm matrix approximation.
\item[Section \ref{spectral_norm_proofs}] Use the reduction to prove random projection and feature selection results.
\item[Section \ref{squish}] Prove $O(k/\epsilon)$ dimension non-oblivious randomized projection result.
\item[Section \ref{logk}] Prove $O(\log k/\epsilon^2)$ random projection result for $(9+\epsilon)$ $k$-means approximation.
\item[Section \ref{applications}] Present example applications of our results to streaming and distributed algorithms.
\end{description}

\section{Preliminaries}\label{prelims}

\subsection{Linear Algebra Basics}

For any $n$ and $d$, consider a matrix $\bv{A} \in \mathbb{R}^{n \times d}$. Let $r = \rank(\bv{A})$. Using a singular value decomposition, we can write $\bv{A} = \bv{U}\bv{\Sigma}\bv{V^\top}$,  where $\bv{U} \in \mathbb{R}^{n \times r}$ and $\bv{V} \in \mathbb{R}^{d \times r}$ have orthogonal columns (the left and right singular vectors of $\bv{A}$) and $\bv{\Sigma} \in \mathbb{R}^{r \times r}$ is a positive diagonal matrix containing the singular values of $\bv{A}$: $\sigma_1 \ge \sigma_2 \ge ... \ge \sigma_r$. The pseudoinverse of $\bv{A}$ is given by $\bv{A}^+ = \bv{V} \bv{\Sigma}^{-1} \bv{U}^\top$.

Let $\bv{\Sigma}_k$ be $\bv{\Sigma}$ with all but its largest $k$ singular values zeroed out. Let $\bv{U}_k$ and $\bv{V}_k$ be $\bv{U}$ and $\bv{V}$ with all but their first $k$ columns zeroed out. For any $k \le r$, $\bv{A}_k = \bv{U}\bv{\Sigma}_k\bv{V^\top} =  \bv{U}_k\bv{\Sigma}_k\bv{V}_k^\top$ is the closest rank $k$ approximation to $\bv{A}$ for any unitarily invariant norm, including the Frobenius norm and spectral norm \cite{1960QJMat1150M}. The squared Frobenius norm is given by $\norm{ \bv{A}}^2_F = \sum_{i,j} \bv{A}_{i,j}^2 = \tr(\bv{A A^\top}) = \sum_{i} \sigma_i^2$. The spectral norm is given by $\norm{\bv{A}}_{2} =\sigma_{1}$.
\begin{align*}
\norm{\bv{A}-\bv{A}_k}_F &= \min_{\bv{B}\mid \rank(\bv{B}) = k} \norm{\bv{A}-\bv{B}}_F \text{ and }\\
\norm{\bv{A}-\bv{A}_k}_2 &= \min_{\bv{B}\mid \rank(\bv{B}) = k} \norm{\bv{A}-\bv{B}}_2.
\end{align*}
We often work with the remainder matrix $\bv{A} - \bv{A}_k$ and label it $\bv{A}_{r \setminus  k}$. 

For any two matrices $\bv{M}$ and $\bv{N}$, $\norm{\bv{MN}}_F \le \norm{\bv{M}}_{F} \norm{\bv{N}}_2 $ and  $\norm{\bv{MN}}_F \le \norm{\bv{N}}_{F} \norm{\bv{M}}_2$. This property is known as \emph{spectral submultiplicativity}. It holds because multiplying by a matrix can scale each row or column, and hence the Frobenius norm, by at most the matrix's spectral norm. Submultiplicativity implies that multiplying by an orthogonal projection matrix (which only has singular values of 0 or 1) can only decrease Frobenius norm, a fact that we will use repeatedly.

If $\bv{M}$ and $\bv{N}$ have the same dimensions and $\bv{MN^\top} = \bv{0}$ then $\norm{\bv{M} + \bv{N}}^2_F = \norm{\bv{M}}^2_F + \norm{\bv{N}}^2_F$. This matrix Pythagorean theorem follows from the fact that $\norm{\bv{M+N}}^2_F = \tr(\bv{(M+N)(M+N)^\top})$. As an example, note that since $\bv{A}_k$ is an orthogonal projection of $\bv{A}$ and $\bv{A}_{r \setminus k}$ is its residual, $\bv{A}_{k}\bv{A}_{r \setminus k}^\top = \bv{0}$. Thus,  $\|\bv{A}_k\|_F^2 + \|\bv{A}_{r \setminus k}\|_F^2 = \|\bv{A}_k+ \bv{A}_{r \setminus k}\|_F^2 = \|\bv{A}\|_F^2$.

For any two symmetric matrices $\bv{M},\bv{N} \in \mathbb{R}^{n \times n}$, $\bv{M} \preceq \bv{N}$ indicates that $\bv{N - M}$ is positive semidefinite -- that is, it has all positive eigenvalues and $\bv{x}^\top (\bv{N}-\bv{M})\bv{x} \ge 0$ for all $\bv{x} \in \mathbb{R}^{n}$. We use $\lambda_i(\bv{M})$ to denote the $i^{\text{th}}$ largest eigenvalue of $\bv{M}$ \emph{in absolute value}.

Finally, we often use $\bv{P}$ to denote an orthogonal projection matrix, which is any matrix that can be written as $\bv{P} = \bv{Q}\bv{Q}^\top$ where $\bv{Q}$ is a matrix with orthonormal columns. Multiplying a matrix by $\bv{P}$ on the left will project its columns to the column span of $\bv{Q}$. If $\bv{Q}$ has just $k$ columns, the projection has rank $k$. Note that $\bv{B}^* = \bv{P}\bv{A}$ minimizes $\norm{\bv{A}-\bv{B}}_F$ amongst all matrices $\bv{B}$ whose columns lie in the column span of $\bv{Q}$ \cite{woodruff2014sketching}.

\subsection{Constrained Low Rank Approximation}
\label{subsec:constrained}
To develop sketching algorithms for $k$-means clustering and low rank approximation, we show that both problems reduce to a general constrained low rank approximation objective. Consider a matrix $\bv{A} \in \mathbb{R}^{n \times d}$ and any set $S$ of rank $k$ orthogonal projection matrices in $\mathbb{R}^{n \times n}$. 
We want to find
\begin{align}\label{projection_objective}
\bv{P^*} = \argmin_{\bv{P} \in S} \norm{\bv{A} - \bv{P}\bv{A}}_F^2.
\end{align}
We often write $\bv{Y} = \bv{I}_{n\times n}-\bv{P}$ and refer to $\norm{\bv{A} - \bv{P}\bv{A}}_F^2 = \norm{\bv{Y}\bv{A}}_F^2$ as the \emph{cost} of the projection $\bv{P}$.

When $S$ is the set of all rank $k$ orthogonal projections, this problem is equivalent to finding the optimal rank $k$ approximation for $\bv{A}$, and is solved by computing $\bv{U}_k$ using an SVD algorithm and setting $\bv{P}^* = \bv{U}_k\bv{U}_k^\top$. In this case, the cost of the optimal projection is $\norm{\bv{A}- \bv{U}_k\bv{U}_k^\top\bv{A}}_F^2 = \norm{\bv{A}_{r \setminus k}}_F^2$. As the optimum cost in the unconstrained case, $\norm{\bv{A}_{r \setminus k}}_F^2$ is a universal lower bound on $\norm{\bv{A} - \bv{P} \bv{A}}_F^2$.

\subsection{$k$-Means Clustering as Constrained Low Rank Approximation}
\label{subsec:k_means_constrained}

Formally, $k$-means clustering asks us to partition $n$ vectors in $\R^d$, $\{\bv{a}_1, \ldots,\bv{a}_n\}$, into $k$ cluster sets, $\{C_1, \ldots, C_k\}$. Let $\bs{\mu}_i$ be the centroid of the vectors in $C_{i}$. Let $\bv{A} \in \mathbb{R}^{n \times d}$ be a data matrix containing our vectors as rows and let $C(\bv{a}_j)$ be the set that vector $\bv{a}_j$ is assigned to. The goal is to minimize the objective function
\begin{align*}
\sum_{i=1}^k \sum_{\bv{a}_{j}\in C_i}\|\bv{a}_{j} - \bs{\mu}_i\|_2^2 = \sum_{j=1}^n \|\bv{a}_{j} - \bs{\mu}_{C(\bv{a}_j)}\|_2^2.
\end{align*}

To see that $k$-means clustering is an instance of general constrained low rank approximation, 
we rely on a linear algebraic formulation of the $k$-means objective that has been used critically in prior work on dimensionality reduction for the problem (see e.g. \cite{DBLP:conf/nips/BoutsidisMD09}). 

For a clustering $C = \{C_1, \ldots, C_k\}$, let $\bv{X}_C \in \mathbb{R}^{n \times k}$ be the \emph{cluster indicator matrix}, with $\bv{X}_C (i,j) = 1/\sqrt{|C_j|}$ if $\bv{a}_i$ is assigned to $C_j$. $\bv{X}_C(i,j) = 0$ otherwise. Thus, $\bv{X}_C\bv{X}_C^\top\bv{A}$ has its $i^{\text{th}}$ row equal to $\bs{\mu}_{C(\bv{a}_i)}$, the center of $\bv{a}_i$'s assigned cluster. So, we can express the $k$-means objective function as:
\begin{align*}
\norm{\bv{A} - \bv{X}_C\bv{X}_C^\top \bv{A}}_F^2 =  \sum_{j=1}^n \|\bv{a}_{j} - \bs{\mu}_{C(\bv{a}_j)}\|_2^2.
\end{align*}
By construction, the columns of $\bv{X}_C$ have disjoint supports and so are orthonormal vectors. Thus $\bv{X}_C\bv{X}_C^\top$ is an orthogonal projection matrix with rank $k$, and $k$-means is just the constrained low rank approximation problem of \eqref{projection_objective} with $S$ as the set of all possible cluster projection matrices $\bv{X}_C\bv{X}_C^\top$.

While the goal of $k$-means is to well approximate each \emph{row} of $\bv{A}$ with its cluster center, this formulation shows that the problem actually amounts to finding an optimal rank $k$ subspace for approximating the \emph{columns} of $\bv{A}$. The choice of subspace is constrained because it must be spanned by the columns of a cluster indicator matrix. 

\section{Projection-Cost Preserving Sketches}\label{error_matrices}

We hope to find an approximately optimal constrained low rank approximation \eqref{projection_objective} for $\bv{A}$ by optimizing $\bv{P}$ (either exactly or approximately) over a sketch $\bv{\tilde A} \in \mathbb{R}^{n \times d'}$ with $d' \ll d$. This approach will certainly work if the cost $\norm{\bv{\tilde A} - \bv{P}\bv{\tilde A}}_F^2$ approximates the cost of $\norm{\bv{A} - \bv{P}\bv{A}}_F^2$ for \emph{any} $\bv{P} \in S$. An even stronger requirement is that $\bv{\tilde A}$ approximates projection-cost for all rank $k$ projections (of which $S$ is a subset). We call such an $\bv{\tilde A}$ a \emph{projection-cost preserving sketch}.

\begin{definition}[Rank $k$ Projection-Cost Preserving Sketch with Two-sided Error]
\label{def:2sidedsketch}
$\bv{\tilde A} \in \R^{n\times d'}$ is a rank $k$ projection-cost preserving sketch of $\bv{A} \in \R^{n\times d}$  with error $0 \leq\epsilon < 1$ if, for all rank $k$ orthogonal projection matrices $\bv{P} \in \R^{n \times n}$,
\begin{align*}
 (1-\epsilon) \norm{\bv{A} - \bv{P}\bv{A}}_F^2 \le \norm{\bv{\tilde A-P\tilde A}}_F^2 + c \le (1+\epsilon)  \norm{\bv{A} - \bv{P}\bv{A}}_F^2,
\end{align*}
for some fixed non-negative constant $c$ that may depend on $\bv{A}$ and $\bv{\tilde A}$ but is independent of $\bv{P}$.
\end{definition}
This definition is equivalent to the $(k,\epsilon)$-coresets of \cite{feldman2013turning} (see their Definition 2). It can be strengthened slightly by requiring a one-sided error bound, which some of our sketching methods will achieve. The tighter bound is required for results that do not have constant factors in the sketch size.

\begin{definition}[Rank $k$ Projection-Cost Preserving Sketch with One-sided Error]
\label{def:1sidedsketch}
$\bv{\tilde A} \in \R^{n\times d'}$ is a rank $k$ projection-cost preserving sketch of $\bv{A} \in \R^{n\times d}$  with one-sided error $0 \leq\epsilon < 1$ if, for all rank $k$ orthogonal projection matrices $\bv{P} \in \R^{n \times n}$,
\begin{align*}
\norm{\bv{A} - \bv{P}\bv{A}}_F^2 \le \norm{\bv{\tilde A-P\tilde A}}_F^2 + c \le (1+\epsilon)  \norm{\bv{A} - \bv{P}\bv{A}}_F^2,
\end{align*}
for some fixed non-negative constant $c$ that may depend on $\bv{A}$ and $\bv{\tilde A}$ but is independent of $\bv{P}$.
\end{definition}

\subsection{Application to Constrained Low Rank Approximation}
\label{subsec:app_to_constrained_low_rank}
It is straightforward to show that a projection-cost preserving sketch is sufficient for approximately optimizing \eqref{projection_objective}, our constrained low rank approximation problem.

\begin{lemma}[Low Rank Approximation via Projection-Cost Preserving Sketches]\label{sketch_approximation} For any $\bv{A}\in \R^{n\times d}$ and any set $S$ of rank $k$ orthogonal projections, let $\bv{P^*} = \argmin_{\bv{P} \in S} \norm{\bv{A} - \bv{P}\bv{ A}}_F^2$. Accordingly, for any $\bv{\tilde A} \in \mathbb{R}^{n \times d'}$, let $\bv{\tilde P}^* = \argmin_{\bv{P} \in S} \norm{\bv{\tilde A} - \bv{P}\bv{\tilde A}}_F^2$. If $\bv{\tilde A}$ is a rank $k$ projection-cost preserving sketch for $\bv{A}$ with error $\epsilon$, then for any $\gamma \ge 1$, 
if $\norm{\bv{\tilde A} - \bv{\tilde P} \bv{\tilde A}}_F^2 \le \gamma \norm{\bv{\tilde A} - \bv{\tilde P}^* \bv{\tilde A}}_F^2$ 
\begin{align*}
\norm{\bv{ A} - \bv{\tilde P} \bv{ A}}_F^2 \le \frac{(1+\epsilon)}{(1-\epsilon)} \cdot \gamma \norm{\bv{A}-\bv{P^* A}}_F^2.
\end{align*}
\end{lemma}

That is, if $\bv{\tilde P}$ is an (approximately) optimal solution for $\bv{\tilde A}$, then it is also approximately optimal for $\bv{A}$. 
\begin{proof}

By optimality of $\bv{\tilde P}^*$ for $\bv{\tilde A}$, $\norm{\bv{\tilde A} - \bv{\tilde P^*}\bv{\tilde A}}_F^2 \leq \norm{\bv{\tilde A} - \bv{P^*}\bv{\tilde A}}_F^2$ and thus,
\begin{align}\label{approx1}
\norm{\bv{\tilde A} - \bv{\tilde P}\bv{\tilde A}}_F^2 \le \gamma \norm{\bv{\tilde A} - \bv{P^*}\bv{\tilde A}}_F^2.
\end{align}
Furthermore, since $\bv{\tilde A}$ is projection-cost preserving, the following two inequalities hold:
\begin{align}\label{approx2}
\norm{\bv{\tilde A} - \bv{P^*}\bv{\tilde A}}_F^2  \le (1+\epsilon) \norm{\bv{A} - \bv{P^*}\bv{ A}}_F^2- c,\\
\label{approx3}
\norm{\bv{\tilde A} - \bv{\tilde P}\bv{\tilde A}}_F^2  \ge (1-\epsilon)\norm{\bv{A} - \bv{\tilde P}\bv{A}}_F^2 - c.
\end{align}
Combining \eqref{approx1},\eqref{approx2}, and \eqref{approx3}, we see that:
\begin{align*}
(1-\epsilon)\norm{\bv{ A} - \bv{\tilde P} \bv{ A}}_F^2 - c &\le (1+\epsilon) \cdot \gamma \norm{\bv{A}-\bv{P^* A}}_F^2 - \gamma c\\
\norm{\bv{ A} - \bv{\tilde P} \bv{ A}}_F^2 &\le \frac{(1+\epsilon)}{(1-\epsilon)} \cdot \gamma \norm{\bv{A}-\bv{P^* A}}_F^2,
\end{align*} 
where the final step is simply the consequence of $c \geq 0$ and $\gamma \ge 1$.
\end{proof}

For any $0 \le \epsilon' < 1$, to achieve a $(1+\epsilon') \gamma$ approximation with Lemma \ref{sketch_approximation}, we just need to set $\epsilon = \frac{\epsilon'}{2+\epsilon'} \ge \frac{\epsilon'}{3}$. Using Definition \ref{def:1sidedsketch} gives a variation that avoids this constant factor adjustment:

\begin{lemma}[Low Rank Approximation via One-sided Error Projection-Cost Preserving Sketches]\label{sketch_approximation_improved} For any $\bv{A}\in \R^{n\times d}$ and any set $S$ of rank $k$ orthogonal projections, let $\bv{P^*} = \argmin_{\bv{P} \in S} \norm{\bv{A} - \bv{P}\bv{ A}}_F^2$. Accordingly, for any $\bv{\tilde A} \in \mathbb{R}^{n \times d'}$, let $\bv{\tilde P}^* = \argmin_{\bv{P} \in S} \norm{\bv{\tilde A} - \bv{P}\bv{\tilde A}}_F^2$. If $\bv{\tilde A}$ is a rank $k$ projection-cost preserving sketch for $\bv{A}$ with one-sided error $\epsilon$, then for any $\gamma \ge 1$, 
if $\norm{\bv{\tilde A} - \bv{\tilde P} \bv{\tilde A}}_F^2 \le \gamma \norm{\bv{\tilde A} - \bv{\tilde P}^* \bv{\tilde A}}_F^2$
\begin{align*}
\norm{\bv{ A} - \bv{\tilde P} \bv{ A}}_F^2 \le (1+\epsilon) \cdot \gamma \norm{\bv{A}-\bv{P^* A}}_F^2.
\end{align*}
\end{lemma}
\begin{proof}
Identical to the proof of Lemma \ref{sketch_approximation} except that \eqref{approx3} can be replaced by $\norm{\bv{\tilde A} - \bv{\tilde P}\bv{\tilde A}}_F^2  \ge \norm{\bv{A} - \bv{\tilde P}\bv{A}}_F^2 - c$, which gives  the result when combined with \eqref{approx1} and \eqref{approx2}.
\end{proof}

\section{Sufficient Conditions}
\label{sec:sufficient_conditions}
With Lemmas \ref{sketch_approximation} and \ref{sketch_approximation_improved} in place, we seek to characterize what sort of sketch suffices for rank $k$ projection-cost preservation. We discuss sufficient conditions that will be used throughout the remainder of the paper. Before giving the full technical analysis, it is helpful to overview our general approach and highlight connections to prior work. 

\subsection{Our Approach}
\label{subsec:our_approach}
Using the notation $\bv{Y} = \bv{I}_{n\times n}-\bv{P}$, we can rewrite the guarantees for Definitions \ref{def:2sidedsketch} and \ref{def:1sidedsketch} as:
\begin{align}
\label{rewritten_guarantee}
 (1-\epsilon) \tr(\bv{Y} \bv{A} \bv{A}^\top \bv{Y}) &\le \tr(\bv{Y} \bv{\tilde A} \bv{\tilde A}^\top \bv{Y}) + c \le (1+\epsilon) \tr(\bv{Y} \bv{A} \bv{A}^\top \bv{Y})\text{, and}\\
\label{rewritten_guarantee_improved}
\tr(\bv{Y}\bv{A} \bv{A}^\top \bv{Y}) &\le \tr(\bv{Y} \bv{\tilde A} \bv{\tilde A}^\top \bv{Y}) + c \le (1+\epsilon) \tr(\bv{Y} \bv{A} \bv{A}^\top \bv{Y}).
\end{align}
Thus, in approximating $\bv{A}$ with $\bv{\tilde A}$, we are really attempting to approximate $\bv{ A} \bv{A}^\top$. 

Furthermore, all of the sketching approaches analyzed in this paper are linear -- i.e. we can always write $\bv{\tilde A} = \bv{AR}$. Suppose our sketching dimension is $m = O(k)$. For an SVD sketch, $\bv{R} = \bv{V}_m$. For a Johnson-Lindenstrauss random projection, $\bv{R}$ is a $d\times m$ random matrix. For a column selection sketch, $\bv{R}$ is a $d\times d$ diagonal matrix with $m$ non-zeros. So, our goal is to show:
\begin{align*}
\tr(\bv{Y} \bv{A} \bv{A}^\top \bv{Y})  \approx \tr(\bv{Y}  \bv{A}\bv{R}\bv{R}^\top \bv{A}^\top \bv{Y}) + c.
\end{align*}
A common trend in prior work has been to attack this analysis by splitting $\bv{A}$ into separate orthogonal components \cite{nphardkmeans,DBLP:journals/corr/abs-1110-2897}. In particular, previous results note that $\bv{A} = \bv{A}_k + \bv{A}_{r \setminus  k}$ and implicitly compare
\begin{align*}
\tr(\bv{Y} \bv{A}\bv{A}^\top \bv{Y}) &= \tr(\bv{Y} \bv{A}_k\bv{A}_k^\top \bv{Y}) + \tr(\bv{Y} \bv{A}_{r \setminus  k}\bv{A}_{r \setminus  k}^\top \bv{Y}) + \tr(\bv{Y} \bv{A}_k\bv{A}_{r \setminus  k}^\top \bv{Y})  + \tr(\bv{Y} \bv{A}_{r \setminus  k}\bv{A}_k^\top \bv{Y}) \\
&= \tr(\bv{Y} \bv{A}_k\bv{A}_k^\top \bv{Y})+ \tr(\bv{Y} \bv{A}_{r \setminus  k}\bv{A}_{r \setminus  k}^\top \bv{Y}) + 0  + 0,
\end{align*}
to
\begin{align*}
\tr(\bv{Y} \bv{A}\bv{R}\bv{R}^\top\bv{A}^\top \bv{Y}) = \tr(\bv{Y} \bv{A}_k\bv{R}\bv{R}^\top\bv{A}_k^\top \bv{Y}) + \tr(\bv{Y} \bv{A}_{r \setminus  k}\bv{R}\bv{R}^\top\bv{A}_{r \setminus  k}^\top \bv{Y})\\ + \tr(\bv{Y} \bv{A}_k\bv{R}\bv{R}^\top\bv{A}_{r \setminus  k}^\top \bv{Y})  + \tr(\bv{Y} \bv{A}_{r \setminus  k}\bv{R}\bv{R}^\top\bv{A}_k^\top \bv{Y}).
\end{align*}
We adopt this same general technique, but make the comparison more explicit and analyze the difference between each of the four terms separately. In Section \ref{splitting_error_tech}, the allowable error in each term will correspond to $\bv{E}_1$, $\bv{E}_2$, $\bv{E}_3$, and $\bv{E}_4$, respectively. 

Additionally, our analysis generalizes the approach by splitting $\bv{A}$ into a wider variety of orthogonal pairs. Our SVD results split $\bv{A} = \bv{A}_{\lceil k/\epsilon\rceil} + \bv{A}_{r \setminus  \lceil k/\epsilon\rceil}$, our random projection results split $\bv{A} = \bv{A}_{2k} + \bv{A}_{r \setminus 2k}$, and our column selection results split $\bv{A} = \bv{AZZ}^\top + \bv{A}(\bv{I-ZZ}^\top)$ for an approximately optimal rank-$k$ projection $\bv{ZZ}^\top$. Finally, our $O(\log k)$ result for $k$-means clustering splits $\bv{A} = \bv{P}^* \bv{A} + (\bv{I-P}^*) \bv{A}$ where $\bv{P}^*$ is the optimal $k$-means projection matrix for $\bv{A}$. 

\subsection{Characterization of Projection-Cost Preserving Sketches}
\label{splitting_error_tech}
Next we formally
analyze what sort of error, $\bv{E} =\bv{\tilde A} \bv{\tilde A}^\top - \bv{A} \bv{A}^\top$, is permissible for a projection-cost preserving sketch. We start by showing how to achieve the stronger guarantee of Definition \ref{def:1sidedsketch} (one-sided error), which will constrain $\bv{E}$ most tightly. We then loosen restrictions on $\bv{E}$ to show conditions that suffice for Definition \ref{def:2sidedsketch} (two-sided error). For ease of notation, write $\bv{C} = \bv{A} \bv{A}^\top$ and $\bv{\tilde C} = \bv{\tilde A} \bv{\tilde A}^\top$.

\begin{lemma}\label{error_improved}
$\bv{\tilde A}$ is a rank $k$ projection-cost preserving sketch with one-sided error $\epsilon$ (i.e. satisfies Definition \ref{def:1sidedsketch}) as long as we can write $\bv{\tilde C} = \bv{C} + \bv{E}$ where
$\bv{E}$ is symmetric, $\bv{E} \preceq \bv{0}$, and $\sum_{i=1}^k |\lambda_i(\bv{E})| \le \epsilon \norm{\bv{A}_{r \setminus k}}_F^2$. Specifically, referring to the guarantee of Equation \ref{rewritten_guarantee_improved},
we show that, for any rank $k$ orthogonal projection $\bv{P}$ and $\bv{Y} = \bv{I-P}$,
\begin{align*}
\tr(\bv{Y} \bv{C} \bv{Y}) \le \tr(\bv{Y} \bv{\tilde C} \bv{Y}) -\tr(\bv{E}) \le (1+\epsilon) \tr(\bv{Y} \bv{C} \bv{Y}).
\end{align*}
\end{lemma}
The general idea of Lemma \ref{error_improved} is fairly simple. Restricting $\bv{E} \preceq \bv{0}$ (which implies $\tr(\bv{E}) \leq 0)$ ensures that the projection independent constant in our sketch is non-negative, which was essential in proving Lemmas \ref{sketch_approximation} and \ref{sketch_approximation_improved}. Then we observe that, since $\bv{P}$ is a rank $k$ projection, any projection dependent error \emph{at worst} depends on the largest $k$ eigenvalues of our error matrix. Since the cost of any rank $k$ projection is at least $\norm{\bv{A}_{r \setminus k}}_F^2$, we need the restriction $\sum_{i=1}^k \left | \lambda_i(\bv{E}) \right | \le \epsilon \norm{\bv{A}_{r \setminus k}}_F^2$ to achieve relative error approximation. 

\begin{proof}
First note that, since $\bv{C} = \bv{\tilde C} - \bv{E}$, by linearity of the trace
\begin{align}
\tr(\bv{Y} \bv{C} \bv{Y}) &= \tr(\bv{Y} \bv{\tilde C} \bv{Y}) - \tr(\bv{Y} \bv{ E} \bv{Y})\nonumber\\
&= \tr(\bv{Y} \bv{\tilde C} \bv{Y}) - \tr(\bv{Y}\bv{E}) \nonumber\\
\label{eq:trace_sep}
&= \tr(\bv{Y} \bv{\tilde C} \bv{Y}) - \tr(\bv{ E}) + \tr(\bv{P E}).
\end{align}
The second step follows from the cyclic property of the trace and the fact that $\bv{Y}^2 = \bv{Y}$ since $\bv{Y}$ is a projection matrix. So, to prove  Lemma \ref{error_improved}, all we have to show is
\begin{align}
\label{improved_reduced_goal}
-\epsilon\tr(\bv{Y} \bv{C} \bv{Y}) \leq \tr(\bv{P E}) \leq 0.
\end{align}
Since $\bv{E}$ is symmetric, let $\bv{v}_1,\ldots,\bv{v}_r$ be the eigenvectors of $\bv{E}$, and write
\begin{align}
\bv{E} &= \sum_{i=1}^r \lambda_i(\bv{E}) \bv{v}_i \bv{v}_i^\top \text{ and thus} \nonumber\\
\label{eq:eigenbreakdown}
\tr(\bv{P}\bv{E}) &= \sum_{i=1}^r \lambda_i(\bv{E}) \tr(\bv{P}\bv{v}_i \bv{v}_i^\top).
\end{align}
For all $i$, $0 \leq \tr(\bv{P}\bv{v}_i \bv{v}_i^\top) \le \norm{\bv{v}_i}_2^2 \le 1$ and $\sum_{i=1}^r \tr(\bv{P}\bv{v}_i \bv{v}_i^\top) \leq \tr(\bv{P}) = k$. Thus, since $\bv{E} \preceq \bv{0}$ and accordingly has all negative eigenvalues, $\sum_{i=1}^r \lambda_i(\bv{E}) \tr(\bv{P}\bv{v}_i \bv{v}_i^\top)$ is minimized when $\tr(\bv{P}\bv{v}_i \bv{v}_i^\top) = 1$ for $\bv{v}_1,\ldots,\bv{v}_k$, the eigenvectors corresponding to $\bv{E}$'s largest magnitude eigenvalues. So,
\begin{align*}
\sum_{i=1}^k \lambda_i(\bv{E}) \leq  \sum_{i=1}^r \lambda_i(\bv{E}) \tr(\bv{P}\bv{v}_i \bv{v}_i^\top) \leq 0.
\end{align*}
The upper bound in Equation \eqref{improved_reduced_goal} follows immediately. The lower bound follows from our requirement that  $\sum_{i=1}^k |\lambda_i(\bv{E})| \le \epsilon\norm{\bv{A}_{r \setminus k}}_F^2$ and the fact that $\norm{\bv{A}_{r \setminus k}}_F^2$ is a universal lower bound on $\tr(\bv{Y} \bv{C} \bv{Y})$ (see Section \ref{subsec:constrained}).
\end{proof}

Lemma \ref{error_improved} is enough to prove that an exact or approximate low rank approximation to $\bv{A}$ gives a sufficient sketch for constrained low rank approximation (see Section \ref{svds}). However, other sketching techniques will introduce a broader class of error matrices, which we handle next.

\begin{lemma}\label{error} 
$\bv{\tilde A}$ is a rank $k$ projection-cost preserving sketch with two-sided error $\epsilon$ (i.e. satisfies Definition \ref{def:2sidedsketch}) as long as we can write $\bv{\tilde C} = \bv{C} + \bv{E}_1 + \bv{E}_2 + \bv{E}_3 + \bv{E}_4$ where
\begin{enumerate}
\item $\bv{E}_1$ is symmetric and $-\epsilon_1 \bv{C} \preceq \bv{E}_1 \preceq \epsilon_1\bv{C}$
\item $\bv{E}_2$ is symmetric, $\sum_{i=1}^k \left | \lambda_i(\bv{E}_2) \right | \le \epsilon_2 \norm{\bv{A}_{r \setminus k}}_F^2$, and $\tr(\bv{E}_2) \le \epsilon_2' \norm{\bv{A}_{r \setminus k}}_F^2$
\item The columns of $\bv{E}_3$ fall in the column span of $\bv{C}$ and $\tr(\bv{E}_3^\top \bv{C}^{+} \bv{E}_3) \le \epsilon_3^2 \norm{\bv{A}_{r  \setminus k}}_F^2$
\item The rows of $\bv{E}_4$ fall in the row span of $\bv{C}$ and $\tr(\bv{E}_4 \bv{C}^{+} \bv{E}_4^\top) \le \epsilon_4^2 \norm{\bv{A}_{r \setminus k}}_F^2$
\end{enumerate}
and $\epsilon_1 + \epsilon_2 + \epsilon_2' + \epsilon_3+ \epsilon_4 = \epsilon$. Specifically, referring to the guarantee in Equation \ref{rewritten_guarantee}, we show that for any rank $k$ orthogonal projection $\bv{P}$ and $\bv{Y} = \bv{I-P}$,
\begin{align*}
(1-\epsilon) \tr(\bv{Y} \bv{C} \bv{Y}) \le \tr(\bv{Y} \bv{\tilde C} \bv{Y}) - \min\{0,\tr(\bv{E}_2)\} \le (1+\epsilon) \tr(\bv{Y} \bv{C} \bv{Y}).
\end{align*}
\end{lemma}

\begin{proof}
Again, by linearity of the trace, note that 
\begin{align}
\label{eq:tildecbrokenout}
\tr(\bv{Y}\bv{\tilde C}\bv{Y}) = \tr(\bv{Y}\bv{C}\bv{Y}) + \tr(\bv{Y}\bv{E}_1\bv{Y}) + \tr(\bv{Y}\bv{E}_2\bv{Y}) + \tr(\bv{Y}\bv{E}_3\bv{Y}) + \tr(\bv{Y}\bv{E}_4\bv{Y}).
\end{align}
We handle each error term separately. Starting with $\bv{E}_1$, note that $\tr(\bv{Y} \bv{E}_1 \bv{Y}) = \sum_{i=1}^n \bv{y}_i^\top \bv{E}_1 \bv{y}_i$ where $\bv{y}_i$ is the $i^\text{th}$ column (equivalently row) of $\bv{Y}$. So, by the spectral bounds on $\bv{E}_1$
\begin{align}\label{e_1}
-\epsilon_1 \tr(\bv{Y} \bv{C} \bv{Y})  \le \tr (\bv{Y} \bv{E}_1 \bv{Y} ) \le \epsilon_1 \tr(\bv{Y} \bv{ C} \bv{Y}). 
\end{align}

$\bv{E}_2$ is analogous to our error matrix from Lemma \ref{error_improved}, but may have both positive and negative eigenvalues since we no longer require $\bv{E}_2 \preceq \bv{0}$ . Again, referring to \eqref{eq:trace_sep}, the goal is to bound $\tr(\bv{Y} \bv{E}_2 \bv{Y}) = \tr(\bv{E}_2) - \tr(\bv{P} \bv{E}_2)$. Using an eigendecomposition as in \eqref{eq:eigenbreakdown}, let $\bv{v}_1,\ldots,\bv{v}_r$ be the eigenvectors of $\bv{E}_2$, and note that
\begin{align*}
|\tr(\bv{P}\bv{E}_2)| = \left | \sum_{i=1}^r \lambda_i(\bv{E}_2) \tr(\bv{P}\bv{v}_i \bv{v}_i^\top) \right | \le \sum_{i=1}^r |\lambda_i(\bv{E}_2)| \tr(\bv{P}\bv{v}_i \bv{v}_i^\top).
\end{align*}
$\sum_{i=1}^r |\lambda_i(\bv{E}_2)| \tr(\bv{P}\bv{v}_i \bv{v}_i^\top)$ is maximized when $\tr(\bv{P}\bv{v}_i \bv{v}_i^\top) = 1$ for $\bv{v}_1,\ldots,\bv{v}_k$. Combined with our requirement that $\sum_{i=1}^k \left | \lambda_i(\bv{E}_2) \right | \le \epsilon_2 \norm{\bv{A}_{r \setminus k}}_F^2$, we see that $|\tr(\bv{P}\bv{E}_2)|  \leq \epsilon_2 \norm{\bv{A}_{r \setminus k}}_F^2$. Accordingly,
\begin{align}
\tr(\bv{E}_2) - \epsilon_2 \norm{\bv{A}_{r \setminus k}}_F^2 \le \tr(\bv{Y} \bv{E}_2 \bv{Y}) &\le  \tr(\bv{E}_2) + \epsilon_2 \norm{\bv{A}_{r \setminus k}}_F^2 \nonumber\\ 
\min\{0,\tr(\bv{E}_2)\} - \epsilon_2 \norm{\bv{A}_{r \setminus k}}_F^2 \le \tr(\bv{Y} \bv{E}_2 \bv{Y}) &\le \min\{0,\tr(\bv{E}_2)\} + (\epsilon_2 + \epsilon_2') \norm{\bv{A}_{r \setminus k}}_F^2 \nonumber \\
\min\{0,\tr(\bv{E}_2)\} - (\epsilon_2 + \epsilon_2') \tr(\bv{Y} \bv{C} \bv{Y}) \le \tr(\bv{Y} \bv{E}_2 \bv{Y}) &\le \min\{0,\tr(\bv{E}_2)\} +(\epsilon_2 + \epsilon_2') \tr(\bv{Y} \bv{C} \bv{Y})\label{e_2}.
\end{align}
The second step follows from the trace bound on $\bv{E}_2$. The last step follows from recalling that $\norm{\bv{A}_{r \setminus k}}_F^2$ is a universal lower bound on $\tr(\bv{Y} \bv{C} \bv{Y})$.

Next, we note that, since $\bv{E}_3$'s columns fall in the column span of $\bv{C}$, $\bv{C}\bv{C}^+ \bv{E}_3 = \bv{E}_3$. Thus,
\begin{align*}
\tr(\bv{Y}\bv{E}_3 \bv{Y}) = \tr(\bv{Y}\bv{E}_3) = \tr \left ((\bv{YC}) \bv{C}^+ (\bv{E}_3)\right ).
\end{align*}
$\langle \bv{M}, \bv{N} \rangle = \tr( \bv{M} \bv{C}^+ \bv{N}^\top)$ is a semi-inner product since $\bv{C} = \bv{A}\bv{A}^\top$, and therefore also $\bv{C}^+$, is positive semidefinite. Thus, by the Cauchy-Schwarz inequality,
\begin{align*}
\left | \tr \left ((\bv{YC}) \bv{C}^+ (\bv{E}_3)\right ) \right | \le \sqrt{\tr(\bv{YC}\bv{C}^+\bv{CY}) \cdot \tr(\bv{E}^\top_3 \bv{C}^+ \bv{E}_3)} \le \epsilon_3 \norm{\bv{A}_{r \setminus k}}_F \cdot \sqrt{\tr(\bv{YCY})}.
\end{align*}
Since $\sqrt{\tr(\bv{YCY})} \ge \norm{\bv{A}_{r \setminus k}}_F$, we conclude that
\begin{align}
\label{e_3}
\left |\tr(\bv{Y}\bv{E}_3 \bv{Y})\right| \leq  \epsilon_3 \cdot \tr(\bv{YCY}).
\end{align}
For $\bv{E}_4$ we make a symmetric argument. 
\begin{align}\label{e_4}
\left | \tr(\bv{Y}\bv{E}_4 \bv{Y}) \right | = \left |\tr \left ((\bv{E}_4) \bv{C}^+ (\bv{CY})\right ) \right | \le \sqrt{\tr(\bv{YCY}) \cdot \tr(\bv{E}_4 \bv{C}^+ \bv{E}_4^\top)} \le \epsilon_4 \cdot \tr(\bv{YCY}).
\end{align} 

Finally, combining equations \eqref{eq:tildecbrokenout}, \eqref{e_1}, \eqref{e_2}, \eqref{e_3}, and \eqref{e_4} and recalling that  $\epsilon_1 + \epsilon_2 + \epsilon_2' + \epsilon_3 + \epsilon_4 = \epsilon$, we have:
\begin{align*}
 (1-\epsilon) \tr(\bv{Y} \bv{C} \bv{Y}) \le \tr(\bv{Y} \bv{\tilde C} \bv{Y}) -  \min\{0,\tr(\bv{E}_2)\} \le (1+\epsilon) \tr(\bv{Y} \bv{C} \bv{Y}).
\end{align*}
\end{proof}
\section{Singular Value Decomposition}\label{svds}

Lemmas \ref{error_improved} and \ref{error} provide a framework for analyzing a variety of projection-cost preserving dimensionality reduction techniques. We start by considering a sketch $\bv{\tilde A}$ that is simply $\bv{A}$ projected onto its top $m = \lceil k/\epsilon \rceil$ singular vectors. As stated, this sketch actually has the same dimensions as $\bv{A}$ -- however, since $\bv{\tilde A} = \bv{A}_{m}$ is simply $\bv{U}_{m} \bv{\Sigma}_{m}$ under rotation, we could actually solve constrained low rank approximation using $\bv{U}_{m} \bv{\Sigma}_{m} = \bv{A} \bv{V}_{m}$ as our data matrix. This form of the sketch has data points of dimension $m = {\lceil k/\epsilon \rceil}$ and can be computed using a truncated SVD algorithm to obtain $\bv{A}$'s top $m$ right singular vectors.

Our analysis is extremely close to \cite{feldman2013turning}, which claims that $O(k/\epsilon^2)$ singular vectors suffice (see their Corollary 4.2). Simply noticing that $k$-means amounts to a constrained low rank approximation problem is enough to tighten their result to   $\lceil k/\epsilon \rceil$. 
In Appendix \ref{lower_bound} we show that $\lceil k/\epsilon \rceil$ is tight -- we cannot take fewer singular vectors and hope to get a $(1+\epsilon)$ approximation in general.

As in \cite{DBLP:journals/corr/abs-1110-2897}, we show that our analysis is robust to imperfection in our singular vector computation. This allows for the use of approximate truncated SVD algorithms, which can be faster than exact methods \cite{tygertImpl}. Randomized SVD algorithms (surveyed in \cite{Halko:2011})  are often highly parallelizable and require few passes over $\bv{A}$, which limits costly memory accesses. In addition to standard Krylov subspace methods like the Lanczos algorithm, asymptotic runtime gains may also be substantial for sparse data matrices .

\subsection{Exact SVD}

\begin{theorem}\label{exact_svd}
Let $m = \lceil k/\epsilon \rceil$. For any $\bv{A} \in \mathbb{R}^{n \times d}$, the sketch $\bv{\tilde A} = \bv{A}_{m}$ satisfies the conditions of Definition \ref{def:1sidedsketch}. Specifically, for any rank $k$ orthogonal projection $\bv{P}$,
\begin{align*}
\norm{\bv{A} - \bv{P}\bv{A}}_F^2 \le \norm{\bv{\tilde A-P\tilde A}}_F^2 + c \le (1+\epsilon)  \norm{\bv{A} - \bv{P}\bv{A}}_F^2.
\end{align*}
\end{theorem}

\begin{proof}
\begin{align*}
\bv{\tilde C} = \bv{\tilde A} \bv{\tilde A}^\top = (\bv{A} - \bv{A}_{r \setminus m})(\bv{A} - \bv{A}_{r  \setminus m})^\top = \bv{ A} \bv{ A}^\top - \bv{A}_{r  \setminus m}\bv{A}_{r  \setminus m}^\top.
\end{align*}
The last equality follows from the fact that $\bv{A}\bv{A}_{r \setminus m}^\top = (\bv{A}_{m}+\bv{A}_{r \setminus m})\bv{A}_{r \setminus m}^\top = \bv{A}_{r \setminus m}\bv{A}_{r \setminus m}^\top$ since the rows of $\bv{A}_{r \setminus m}$ and $\bv{A}_{m}$ lie in orthogonal subspaces and so $\bv{A}_{m}\bv{A}_{r \setminus m}^\top = \bv{0}$.
Now, we simply apply Lemma \ref{error_improved}, setting $\bv{E} = -\bv{A}_{r \setminus m}\bv{A}_{r \setminus m}^\top$. We know that $\bv{\tilde C} = \bv{C} + \bv{E}$, $\bv{E}$ is symmetric, and $\bv{E} \preceq \bv{0}$ since $\bv{A}_{r \setminus m}\bv{A}_{r \setminus m}^\top$ is positive semidefinite. Finally,
\begin{align}\label{singular_value_bound}
\sum_{i=1}^k |\lambda_i(\bv{E})| = \sum_{i=1}^k \sigma^2_i(\bv{A}_{r \setminus m}) =  \sum_{i=m +1}^{m +k} \sigma^2_i(\bv{A}) \le \epsilon \norm{\bv{A}_{r \setminus k}}_F^2. 
\end{align}
The final inequality follows from the fact that 
\begin{align}\label{svd_head_bound}
 \norm{\bv{A}_{r \setminus k}}_F^2  = \sum_{i=k+1}^n \sigma_i^2(\bv{A}) \ge \sum_{i=k+1}^{m+k} \sigma_i^2(\bv{A}) \ge \frac{1}{\epsilon} \sum_{i=m+1}^{m+k} \sigma_i^2(\bv{A})
\end{align}
since the last sum contains just the smallest $k$ terms of the previous sum, which has $m = \lceil k/\epsilon \rceil$ terms in total.
So, by Lemma \ref{error_improved}, we have:
\begin{align*}
\norm{\bv{A} - \bv{P}\bv{A}}_F^2 \le \norm{\bv{\tilde A-P\tilde A}}_F^2 + c \le (1+\epsilon)  \norm{\bv{A} - \bv{P}\bv{A}}_F^2.
\end{align*}
\end{proof}

Note that, in practice, it may be possible to set $m \ll \lceil k/\epsilon \rceil$. Specifically, $\lceil k/\epsilon \rceil$ singular vectors are only required for the condition of Equation \ref{singular_value_bound},
\begin{align*}
\sum_{i=m +1}^{m +k} \sigma^2_i(\bv{A}) \le \epsilon \norm{\bv{A}_{r \setminus k}}_F^2,
\end{align*} 
when the top $\lceil k/\epsilon \rceil$ singular values of $\bv{A}$ are all equal. If the spectrum of $\bv{A}$ decays, the equation will hold for a smaller $m$. Furthermore, it is easy to check the condition by iteratively computing the singular values of $\bv{A}$ and stopping once a sufficiently high $m$ is found.

Finally, note that $\norm{\bv{\tilde A-P\tilde A}}_F^2 = \norm{(\bv{I-P}) \bv{U}_m\bs{\Sigma}_m\bv{V}_m^\top}_F^2 = \norm{(\bv{I-P}) \bv{U}_m\bs{\Sigma}_m}_F^2$ since $\bv{V}_m^\top$ is orthonormal. So, as claimed, the sketch $\bv{\tilde A} = \bv{U}_m\bs{\Sigma}_m \in \R^{n\times  \lceil k/\epsilon \rceil}$ also satisfies Definition \ref{def:1sidedsketch}.

%

\subsection{Approximate SVD}
\label{subsec:approx_svd}

Next we claim that any approximately optimal set of top singular vectors suffices for sketching $\bv{A}$.
\begin{theorem}\label{approx_svd}
Let $m = \lceil k/\epsilon \rceil$. For any $\bv{A} \in \mathbb{R}^{n \times d}$ and any orthonormal matrix $\bv{Z} \in \mathbb{R}^{d \times m}$ satisfying $\norm{\bv{A}-\bv{A}\bv{Z}\bv{Z}^\top}_F^2 \le (1+\epsilon') \norm{\bv{A}_{r \setminus m}}_F^2$, the sketch $\bv{\tilde A} = \bv{A}\bv{Z}\bv{Z}^\top$ satisfies the conditions of Definition \ref{def:1sidedsketch}. Specifically, for all rank $k$ orthogonal projections $\bv{P}$,
\begin{align*}
\norm{\bv{A} - \bv{P}\bv{A}}_F^2 \le \norm{\bv{\tilde A-P\tilde A}}_F^2 + c \le (1+\epsilon+\epsilon')  \norm{\bv{A} - \bv{P}\bv{A}}_F^2.
\end{align*}
\end{theorem}
In recent years, this sort of relative error approximation to the SVD has become standard \cite{sarlos2006improved,Halko:2011}. 
Additionally, note that this theorem implies that the sketch $\bv{A}\bv{Z} \in \R^{n\times  \lceil k/\epsilon \rceil}$ also satisfies Definition  \ref{def:1sidedsketch}. The proof of Theorem \ref{approx_svd} is included in Appendix \ref{approx_svd_proofs}.

\subsection{General Low Rank Approximation}

Finally, we consider an even more general case when $\bv{\tilde A}$ is a good low rank approximation of $\bv{A}$ but may not actually be a row projection of $\bv{A}$ -- i.e. $\bv{\tilde A}$ doesn't necessarily take the form $\bv{AZZ}^\top$. This is the sort of sketch obtained, for example, by the randomized low rank approximation result in \cite{Clarkson:2013} (see Theorem 47). Note that  \cite{Clarkson:2013} still returns a decomposition of the computed sketch, $\bv{\tilde A} = \bv{LDW}^\top$, where $\bv{L}$ and $\bv{W}$ have orthonormal columns and $\bv{D}$ is a $k\times k$ diagonal matrix. Thus, by using $\bv{LD}$, which has just $m$ columns, it is still possible to solve $k$-means (or some other constrained low rank approximation problem) on a matrix that is much smaller than $\bv{A}$.

 \begin{theorem}\label{general_svd}
Let $m = \lceil k/\epsilon \rceil$. For any $\bv{A} \in \mathbb{R}^{n \times d}$ and any $\bv{\tilde A} \in \mathbb{R}^{n \times d}$ with $rank(\bv{\tilde A}) = m$ satisfying $\norm{\bv{A}-\bv{\tilde A}}_F^2 \le (1+(\epsilon')^2) \norm{\bv{A}_{r \setminus m}}_F^2$, the sketch $\bv{\tilde A}$ satisfies the conditions of Definition \ref{def:2sidedsketch}. Specifically, for all rank $k$ orthogonal projections $\bv{P}$,
\begin{align*}
 (1-2\epsilon') \norm{\bv{A} - \bv{P}\bv{A}}_F^2 \le \norm{\bv{\tilde A-P\tilde A}}_F^2 + c \le (1+2\epsilon+5\epsilon')  \norm{\bv{A} - \bv{P}\bv{A}}_F^2.
\end{align*}
\end{theorem}

Generally, the result follows from noting that any good low rank approximation to $\bv{A}$ cannot be far from an actual   rank $k$ projection of $\bv{A}$. Our proof is included in Appendix \ref{approx_svd_proofs}.

\section{Reduction to Spectral Norm Matrix Approximation}
\label{block_section}

To prove our column selection and random projection results, we rely on a reduction from the requirements of Lemma \ref{error} to \emph{spectral norm matrix approximation}. For column selection and random projection, we can always write $\bv{\tilde A} = \bv{A}\bv{R}$, where $\bv{R}\in \R^{d\times m}$ is either a diagonal matrix that selects and reweights columns of $\bv{A}$ or a random Johnson-Lindenstrauss matrix. In order to simplify our proofs we wish to construct a new matrix $\bv{B}$ such that, along with a few other conditions,
\begin{align*}
\norm{\bv{BR}\bv{R}^\top\bv{B}^\top - \bv{B}\bv{B}^\top}_2 < \epsilon
\end{align*}
 implies that $\bv{\tilde A} = \bv{A}\bv{R}$ satisfies the conditions of Lemma \ref{error}. Specifically we show:
\begin{lemma}
\label{blocklem}
Suppose that, for $m \le 2k$, we have some $\bv{Z} \in \mathbb{R}^{d \times m}$ with orthonormal columns satisfying $\norm{\bv{A}-\bv{A}\bv{Z}\bv{Z}^\top}_F^2 \le 2 \norm{\bv{A}_{r \setminus k}}_F^2$ and $\norm{\bv{A}-\bv{A}\bv{Z}\bv{Z}^\top}^2_2 \le \frac{2}{k} \norm{\bv{A}_{r \setminus k}}_F^2$. Set $\bv{B} \in \mathbb{R}^{(n+m) \times d}$ to have $\bv{B}_1 = \bv{Z}^\top$ as its first $m$ rows and $\bv{B}_2 = \frac{\sqrt{k}}{\norm{\bv{A}_{r \setminus k}}_F} \cdot(\bv{A-AZZ}^\top)$ as its lower $n$ rows. Then $1 \le \norm{\bv{BB}^\top}_2 \le 2$, $\tr(\bv{B}\bv{B}^\top) \le 3k$, and $\tr(\bv{B}_2\bv{B}_2^\top) \le 2k$. Furthermore, if
\begin{align}\label{spectral_cond}
\norm{\bv{BR}\bv{R}^\top\bv{B}^\top - \bv{B}\bv{B}^\top}_2 < \epsilon
\end{align}
and 
\begin{align}\label{trace_cond}
\tr(\bv{B}_2 \bv{RR}^\top \bv{B}_2^\top) - \tr(\bv{B}_2\bv{B}_2^\top) \le \epsilon k,
\end{align}
then $\bv{\tilde A} = \bv{AR}$ satisfies the conditions of Lemma \ref{error} with error $ 6\epsilon$.
\end{lemma}

Note that the construction of $\bv{B}$ is really an approach to splitting $\bv{A}$ into orthogonal pairs as described in Section \ref{subsec:our_approach}. The conditions on $\bv{Z}$ ensure that $\bv{AZ}\bv{Z}^\top$ is a good low rank approximation for $\bv{A}$ in both the Frobenius norm and spectral norm sense. We could simply define $\bv{B}$ with $\bv{Z} = \bv{V}_{2k}$, the top right singular vectors of $\bv{A}$. In fact, this is what we will do for our random projection result. However, in order to compute sampling probabilities for column selection, we will need to compute $\bv{Z}$ explicitly and so want the flexibility of using an approximate SVD algorithm. 

\begin{proof} 
We first show that $1\le \norm{\bv{BB}^\top}_2 \le 2$. Notice that $\bv{B}_1 \bv{B}_2^\top = \bv{0}$, so $\bv{BB}^\top$ is a block diagonal matrix with an upper left block equal to $\bv{B}_1\bv{B}_1^\top = \bv{I}$ and lower right block equal to $\bv{B}_2 \bv{B}_2^\top$. 
The spectral norm of the upper left block is $1$. By our spectral norm bound on $\bv{A}-\bv{A}\bv{ZZ}^\top$, $\norm{\bv{B}_2\bv{B}_2^\top}_2 \le \frac{2}{k} \norm{\bv{A}_{r \setminus k}}_F^2 \frac{k}{\norm{\bv{A}_{r \setminus k}}_F^2} = 2$, giving us the upper bound for $\bv{BB}^\top$. Additionally,  $\tr(\bv{B}_2\bv{B}_2^\top) \le \frac{k}{\norm{\bv{A}_{r \setminus k}}_F^2} \norm{\bv{A}-\bv{A}\bv{Z}\bv{Z}^\top}_F^2 \le 2k$ by our Frobenius norm condition on $\bv{A}-\bv{A}\bv{Z}\bv{Z}^\top$. Finally, $\tr(\bv{B}\bv{B}^\top) = \tr(\bv{B}_1\bv{B}_1^\top) + \tr(\bv{B}_2\bv{B}_2^\top) \le 3k$.

We now proceed to the main reduction. Start by setting $\bv{E} = \bv{\tilde C} - \bv{C} = \bv{AR}\bv{R}^\top \bv{A}^\top - \bv{AA}^\top$.
Now, choose $\bv{W}_1 \in \mathbb{R}^{n \times (n+m)}$ such that $\bv{W}_1\bv{B} = \bv{A}\bv{ZZ}^\top$. Note that $\bv{W}_1$ has all columns other than its first $m$ as zero, since reconstructing $\bv{A}\bv{ZZ}^\top$ only requires recombining rows of $\bv{B}_1 = \bv{Z}^\top$. Set $\bv{W}_2 \in \mathbb{R}^{n \times (n+m)}$ to have its first $m$ columns zero and its next $n$ columns as the $n \times n$ identity matrix multiplied by $\frac{\norm{\bv{A}_{r \setminus k}}_F}{\sqrt{k}}$. This insures that $\bv{W}_2 \bv{B} =\frac{\norm{\bv{A}_{r \setminus k}}_F}{\sqrt{k}} \bv{B}_2 = \bv{A} - \bv{AZZ}^\top$. So, $\bv{A} = \bv{W}_1 \bv{B} + \bv{W}_2 \bv{B}$ and we can rewrite:
\begin{align*}
\bv{E} = (\bv{W}_1 \bv{B} \bv{R} \bv{R}^\top \bv{B}^\top \bv{W}_1^\top - \bv{W}_1 \bv{B} \bv{B}^\top \bv{W}_1^\top) + 
(\bv{W}_2 \bv{B} \bv{R} \bv{R}^\top \bv{B}^\top \bv{W}_2^\top - \bv{W}_2 \bv{B} \bv{B}^\top \bv{W}_2^\top) +\\
 (\bv{W}_1 \bv{B} \bv{R} \bv{R}^\top \bv{B}^\top \bv{W}_2^\top - \bv{W}_1 \bv{B} \bv{B}^\top \bv{W}_2^\top) + 
(\bv{W}_2 \bv{B} \bv{R} \bv{R}^\top \bv{B}^\top \bv{W}_1^\top - \bv{W}_2 \bv{B} \bv{B}^\top \bv{W}_1^\top)
\end{align*}

We consider each term of this sum separately, showing that each corresponds to one of the allowed error terms from Lemma \ref{error}. Set $\bv{E}_1 = (\bv{W}_1 \bv{B} \bv{R} \bv{R}^\top \bv{B}^\top \bv{W}_1^\top - \bv{W}_1 \bv{B} \bv{B}^\top \bv{W}_1^\top)$. Clearly $\bv{E}_1$ is symmetric. If, as required, $\norm{\bv{BR}\bv{R}^\top\bv{B}^\top - \bv{B}\bv{B}^\top}_2 < \epsilon$,  $-\epsilon \bv{I} \preceq (\bv{BR}\bv{R}^\top\bv{B}^\top - \bv{B}\bv{B}^\top) \preceq \epsilon \bv{I}$ so $-\epsilon \bv{W}_1 \bv{W}_1^\top \preceq \bv{E}_1 \preceq \epsilon \bv{W}_1 \bv{W}_1^\top$. Furthermore, $\bv{W}_1 \bv{BB}^\top \bv{W}_1^\top = \bv{A}\bv{ZZ}^\top \bv{ZZ}^\top \bv{A}^\top \preceq \bv{A}\bv{A}^\top = \bv{C}$. Since $\bv{W}_1$ is all zeros except in its first $m$ columns and since $\bv{B}_1 \bv{B}_1^\top = \bv{I}$, $\bv{W}_1 \bv{W}_1^\top = \bv{W}_1 \bv{BB}^\top \bv{W}_1^\top$. This gives us:
\begin{align}\label{w_bound}
\bv{W}_1 \bv{W}_1^\top = \bv{W}_1 \bv{BB}^\top \bv{W}_1^\top \preceq \bv{C}.
\end{align}
So overall we have:
\begin{align}\label{block_1}
-\epsilon \bv{C} \preceq \bv{E}_1 \preceq \epsilon \bv{C},
\end{align}
satisfying the error conditions of Lemma \ref{error}.

Next, set $\bv{E}_2 = (\bv{W}_2 \bv{B} \bv{R} \bv{R}^\top \bv{B}^\top \bv{W}_2^\top - \bv{W}_2 \bv{B} \bv{B}^\top \bv{W}_2^\top)$. Again, $\bv{E}_2$ is symmetric and 
\begin{align}\label{block_2_1}
\tr(\bv{E}_2)  = \frac{\norm{\bv{A}_{r \setminus k}}_F^2}{k}\tr(\bv{B}_2 \bv{R} \bv{R}^\top \bv{B}_2^\top - \bv{B}_2 \bv{B}_2^\top) \le \epsilon \norm{\bv{A}_{r \setminus k}}_F^2
\end{align}
by condition \eqref{trace_cond}. Furthermore,
\begin{align}
\sum_{i=1}^k \left | \lambda_i(\bv{E}_2) \right | &\le k \cdot |\lambda_1 (\bv{E}_2)| \nonumber \\
&\le k \cdot \frac{\norm{\bv{A}_{r \setminus k}}_F^2}{k} |\lambda_1(\bv{B}_2 \bv{R} \bv{R}^\top \bv{B}_2^\top - \bv{B}_2 \bv{B}_2^\top)| \nonumber\\
&\le \norm{\bv{A}_{r \setminus k}}_F^2 \cdot |\lambda_1(\bv{B} \bv{R} \bv{R}^\top \bv{B}^\top - \bv{B} \bv{B}^\top)|\nonumber\\
&\le \epsilon \norm{\bv{A}_{r \setminus k}}_F^2\label{block_2_2}
\end{align}
by condition \eqref{spectral_cond}. So $\bv{E}_2$ also satisfies the conditions of Lemma \ref{error}.

Next, set $\bv{E}_3 =  (\bv{W}_1 \bv{B} \bv{R} \bv{R}^\top \bv{B}^\top \bv{W}_2^\top - \bv{W}_1 \bv{B} \bv{B}^\top \bv{W}_2^\top)$. The columns of $\bv{E}_3$ are in the column span of $\bv{W}_1\bv{B} = \bv{A}\bv{Z}\bv{Z}^\top$, and so in the column span of $\bv{C}$. 
Now:
\begin{align*}
\bv{E}_3^\top \bv{C}^+ \bv{E}_3 = \bv{W}_2 (\bv{B} \bv{R} \bv{R}^\top \bv{B}^\top - \bv{B}\bv{B}^\top) \bv{W}_1^\top \bv{C}^+ \bv{W}_1 (\bv{B} \bv{R} \bv{R}^\top \bv{B}^\top - \bv{B}\bv{B}^\top) \bv{W}_2^\top.
\end{align*}
$\bv{W}_1 \bv{W}_1^\top \preceq \bv{C}$ by \eqref{w_bound}, so $\bv{W}_1^\top \bv{C}^+ \bv{W}_1 \preceq \bv{I}$. So:
\begin{align*}
\bv{E}_3^\top \bv{C}^+ \bv{E}_3 &\preceq \bv{W}_2 (\bv{B} \bv{R} \bv{R}^\top \bv{B}^\top - \bv{B}\bv{B}^\top)^2 \bv{W}_2^\top
\end{align*}
which gives:
\begin{align*}
\norm{\bv{E}_3^\top \bv{C}^+ \bv{E}_3}_2 &\le \norm{\bv{W}_2 (\bv{B} \bv{R} \bv{R}^\top \bv{B}^\top - \bv{B}\bv{B}^\top)^2 \bv{W}_2^\top}_2 \le \frac{\norm{\bv{A}_{r \setminus k}}_F^2}{k} \norm{(\bv{B} \bv{R} \bv{R}^\top \bv{B}^\top - \bv{B}\bv{B}^\top)^2}_2 \le \epsilon^2 \frac{\norm{\bv{A}_{r \setminus k}}_F^2}{k}
\end{align*}
by condition \eqref{spectral_cond}.
Now, $\bv{E}_3$ and hence $\bv{E}_3^\top \bv{C}^+ \bv{E}_3$ only have rank $m \le 2k$ so \begin{align}\label{block_3}
\tr(\bv{E}_3^\top \bv{C}^+ \bv{E}_3) \le 2\epsilon^2 \norm{\bv{A}_{r \setminus k}}^2_F.
\end{align}

Finally, we set $\bv{E}_4 =  (\bv{W}_2 \bv{B} \bv{R} \bv{R}^\top \bv{B}^\top \bv{W}_1^\top - \bv{W}_2 \bv{B} \bv{B}^\top \bv{W}_1^\top) = \bv{E}_3^\top$ and thus immediately have:
\begin{align}\label{block_4}
\tr(\bv{E}_4 \bv{C}^+ \bv{E}_4^\top) \le 2\epsilon^2 \norm{\bv{A}_{r \setminus k}}^2_F.
\end{align}
Together, \eqref{block_1}, \eqref{block_2_1}, \eqref{block_2_2}, \eqref{block_3}, and \eqref{block_4} ensure that $\bv{\tilde A} = \bv{AR}$ satisfies Lemma \ref{error} with error $3\epsilon + 2\sqrt{2} \epsilon \le 6\epsilon$.

\end{proof}
\section{Random Projection and Feature Selection}\label{spectral_norm_proofs}
The reduction in Lemma \ref{blocklem} reduces the problem of finding a projection-cost preserving sketch to well understood matrix sketching guarantees -- subspace embedding \eqref{spectral_cond} and trace preservation \eqref{trace_cond}. A variety of known sketching techniques achieve the error bounds required, including several families of
\emph{subspace embedding} matrices which are referred to as Johnson-Lindenstrauss or random projection matrices throughout this paper. These families are listed alongside randomized column sampling and deterministic column selection sketches below. Note that, to better match previous writing in this area, the matrix matrix $\bv M$ given below will correspond to the transpose of $\bv B$ in Lemma \ref{blocklem}. 

\begin{lemma}
\label{matapprox}
Let $\bv M$ be a matrix with $q$ rows, $\norm{\bv M^\top \bv M}_2 \le 1$, and $\frac{\tr(\bv M^\top \bv M)}{\norm{\bv M^\top \bv M}_2} \leq k$. Suppose $\bv{R}$ is a sketch drawn from any of the following probability distributions of matrices.  Then, for any $\epsilon < 1$ and $\delta < 1/2$, $\norm{\bv{M}^\top \bv{R}^\top \bv R \bv M - \bv{M}^\top \bv M}_2 \leq \epsilon$ and $\left|\tr(\bv{M}^\top \bv{R}^\top \bv R \bv M) - \tr(\bv{M}^\top \bv M)\right| \leq \epsilon k$ with probability at least $1-\delta$.
\begin{enumerate}
\item
$\bv R$ a dense Johnson-Lindenstrauss matrix: a matrix with $q$ columns and $d' = O \left ( \frac{k + \log(1/\delta)}{\epsilon^2} \right )$ rows, with each element chosen independently and uniformly $\pm \sqrt{\frac{1}{d'}}$ \cite{Achlioptas:2003}.  Additionally, the same matrix family except with elements only $O(\log(k / \delta))$-independent \cite{clarkson2009numerical}.

\item
$\bv R$ a fully sparse embedding matrix: a matrix with $q$ columns and $d' = O \left ( \frac{k^2}{\epsilon^2 \delta} \right )$ rows, where each column has a single $\pm 1$ in a random position (sign and position chosen uniformly and independently).  Additionally, the same matrix family except where the position and sign for each column are determined by a 4-independent hash function \cite{Clarkson:2013,meng2013low, DBLP:conf/focs/NelsonN13}.

\item
$\bv R$ an OSNAP sparse subspace embedding matrix \cite{DBLP:conf/focs/NelsonN13}.

\item
$\bv R$ a diagonal matrix that samples and reweights $d' = O \left (\frac{k\log(k / \delta)}{\epsilon^2} \right)$ rows of $\bv M$, selecting each with probability proportional to $\norm{\bv M_i}_2^2$ and reweighting by the inverse probability. Alternatively, $\bv R$ that samples $O \left (\frac{\sum_i t_i \log(\sum_i t_i / \delta)}{\epsilon^2} \right)$ rows of $\bv{M}$ each with probability proportional $t_i$, where $t_i \ge \norm{\bv M_i}_2^2$ for all $i$
\cite{stablechernoff}.

\item
$\bv R$ a `BSS matrix': a deterministic diagonal matrix generated by a polynomial time algorithm that selects and reweights $d' = O\left( \frac{k}{\epsilon^2} \right)$ rows of $\bv{M}$ \cite{batson2012twice,jelaniDiscussion}.
\end{enumerate}
\end{lemma}

Lemma \ref{matapprox} requires that $\bv M$ has \emph{stable rank} $\frac{\norm{\bv M}_F^2}{\norm {\bv M}_2^2} \le k$. It is well known that if $\bv M$ has \emph{rank} $\leq k$, the $\norm{\bv{M}^\top \bv{R}^\top \bv R \bv M - \bv{M}^\top \bv M}_2 \leq \epsilon$ bound holds for families \emph{1}, \emph{2}, and \emph{3} because they are all subspace embedding matrices. It can be shown that the relaxed stable rank guarantee is sufficient as well \cite{jelaniDiscussion}. We include an alternative proof for families \emph{1}, \emph{2}, and \emph{3} under Theorem \ref{rp_theorem} that
gives a slightly worse $\delta$ dependence for some constructions but does not rely on these stable rank results.

For family \emph{4}, the $\norm{\bv{M}^\top \bv{R}^\top \bv R \bv M - \bv{M}^\top \bv M}_2 \leq \epsilon$ result follows from Example 4.3 in \cite{stablechernoff}. Family \emph{5} uses a variation on the algorithm introduced in \cite{batson2012twice} and extended in \cite{jelaniDiscussion} to the stable rank case.

Since $\norm{\bv{M}^\top \bv{M}}_2 \le 1$, our stable rank requirement ensures that $\tr(\bv M^\top \bv M) = \norm{\bv M}_F^2 \le k$. Thus, the $\left|\tr(\bv{M}^\top \bv{R}^\top \bv R \bv M) - \tr(\bv{M}^\top \bv M)\right| \leq \epsilon k$ bound holds as long as $\left|\norm{\bv R \bv M}_F^2 - \norm{\bv M}_F^2\right| \leq \epsilon \norm{\bv M}_F^2$. This Frobenius norm bound is standard for embedding matrices and can be proven via the JL-moment property (see Lemma 2.6 in \cite{clarkson2009numerical} or Problem 2(c) in \cite{jelanipset}). 
For family \emph{1}, a proof of the required moment bounds can be found in Lemma 2.7 of \cite{clarkson2009numerical}. 
For family \emph{2} see Remark 23 in \cite{kane2014sparser}. 
For family \emph{3} see Section 6 in \cite{kane2014sparser}. 
For family \emph{4}, the $\left|\norm{\bv R \bv M}_F^2 - \norm{\bv M}_F^2\right| \leq \epsilon \norm{\bv M}_F^2$ bound follows from applying the Chernoff bound. 
For family \emph{5}, the Frobenius norm condition is met by computing $\bv{R}$ using a matrix $\bv M'$. $\bv M'$ is formed by appending a column to $\bv{M}$ whose $i^{th}$ entry is equal to $\norm{\bv{M}^i}_2$ -- the $\ell_2$ norm of the $i^\text{th}$ row of $\bv{M}$. With this column appended, if $\bv{R}$ preserves the spectral norm of $\bv{M'}^\top \bv M'$ up to $\epsilon$ error, it must also preserve the spectral norm of $\bv M^\top \bv M$. Additionally, it must preserve $(\bv{M'}^\top \bv M')_{ii}  = \norm{\bv{M}}_F^2$. The stable rank condition still holds for $\bv M'$ with $k' = 2k$ since appending the column doubles the squared Frobenius norm and does not decrease spectral norm. 


To apply the matrix families from Lemma \ref{matapprox} to Lemma~\ref{blocklem},  we first set $\bv M$ to $\frac{1}{2} \bv B^\top$ and use the sketch matrix $\bv R^\top$. Applying Lemma \ref{matapprox} with $\epsilon' = \epsilon/4$ gives requirement \eqref{spectral_cond} with probability $1-\delta$.  
For families \emph{1}, \emph{2}, and \emph{3}, \eqref{trace_cond} follows from applying Lemma \ref{matapprox} separately with $\bv M = \frac{1}{2}\bv{B}_2^\top$ and $\epsilon' = \epsilon/4$. 
For family \emph{4}, the trace condition follows from noting that sampling probabilities computed using $\bv{B}$ upper bound the correct probabilities for $\bv{B}_2$ and are thus sufficient. 
For family \emph{5}, to get the trace condition we can use the procedure described above, except $\bv B'$ has a row with the column norms of $\bv{B}_2$ as its entries, rather than the column norms of $\bv B$.

\subsection{Random Projection}
Since the first three matrix families listed are all oblivious (do not depend on $\bv{M}$) we can apply Lemma \ref{blocklem} with any suitable $\bv B$, including the one coming from the exact SVD with $\bv Z = \bv V_{2k}$. Note that $\bv B$ \emph{does not} need to be computed at all to apply these oblivious reductions -- it is purely for the analysis. This gives our main random projection result:
\begin{theorem}\label{rp_theorem}
Let $\bv R \in \R^{d' \times d}$ be drawn from any of the first three matrix families from Lemma \ref{matapprox}.  Then, for any matrix $\bv A \in \mathbb{R}^{n \times d}$, with probability at least $1-O(\delta)$, $\bv A \bv R^\top$ is a rank $k$ projection-cost preserving sketch of $\bv A$ (i.e. satisfies Definition \ref{def:2sidedsketch}) with error $O(\epsilon)$.
\end{theorem}

Family \emph{1} gives oblivious reduction to $O(k/\epsilon^2)$ dimensions, while family \emph{2} achieves $O(k^2/\epsilon^2)$ dimensions with the advantage of being faster to apply to $\bv A$, especially when our data is sparse.  Family \emph{3} allows a tradeoff between output dimension and computational cost.

A simple proof of Theorem \ref{rp_theorem} can be obtained that avoids work in \cite{jelaniDiscussion} and only depends on more well establish Johnson-Lindenstrauss properties. We set $\bv{Z} = \bv{V}_{k}$ and bound the error terms from Lemma \ref{blocklem} directly (without going through Lemma \ref{matapprox}). The bound on $\bv{E}_1$ \eqref{block_1} follows from noting that $\bv{W}_1 \bv{B} = \bv{A}\bv{V}_k\bv{V}_k^\top$ only has rank $k$. Thus, we can apply the fact that families \emph{1}, \emph{2}, and \emph{3} are subspace embeddings to claim that $\tr(\bv{W}_1 \bv{B} \bv{R} \bv{R}^\top \bv{B}^\top \bv{W}_1^\top - \bv{W}_1 \bv{B} \bv{B}^\top \bv{W}_1^\top) \leq \epsilon\tr(\bv{W}_1 \bv{B} \bv{B}^\top \bv{W}_1^\top)$. 

The bound on $\bv{E}_2$ \eqref{block_2_2} follows from first noting that, since we set  $\bv{Z} = \bv{V}_{k}$, $\bv{E}_2 = (\bv{A}_{r \setminus k} \bv{R} \bv{R}^\top \bv{A}_{r \setminus k}^\top - \bv{A}_{r \setminus k} \bv{A}_{r \setminus k}^\top)$. Applying Theorem 21 of \cite{kane2014sparser} (approximate matrix multiplication) along with the referenced JL-moment bounds for our first three families gives $\norm{\bv{E}_2}_F \le \frac{\epsilon}{\sqrt{k} } \norm{\bv{A}_{r \setminus k}}_F^2$. Since $\sum_{i=1}^k \left | \lambda_i(\bv{E}_2) \right | \le \sqrt{k} \norm{\bv{E}_2}_F$, \eqref{block_2_2} follows. Note that \eqref{block_2_1} did not require the stable rank generalization, so we do not need any modified analysis.

Finally, the bounds on $\bv{E}_3$ and $\bv{E}_4$, \eqref{block_3} and \eqref{block_4}, follow from the fact that:
\begin{align*}
\tr(\bv{E}_3^\top \bv{C}^+ \bv{E}_3) = \norm{\bv{\Sigma}^{-1} \bv{U}^\top (\bv{W}_1 \bv{B} \bv{R} \bv{R}^\top \bv{B}^\top \bv{W}_2^\top - \bv{W}_1 \bv{B} \bv{B}^\top \bv{W}_2^\top)}_F^2 = \norm{\bv{V}_k \bv{RR}^\top \bv{A}_{r \setminus k}^\top}_F^2 \le \epsilon^2 \norm{\bv{A}_{r \setminus k}}_F^2
\end{align*}
again by Theorem 21 of \cite{kane2014sparser} and the fact that $\norm{\bv{V}_k}_F^2 = k$. In both cases, we apply the approximate matrix multiplication result with error $\epsilon/\sqrt{k}$. For family \emph{1}, the required moment bound needs a sketch with dimension $O\left(\frac{k\log(1/\delta)}{\epsilon^2}\right)$ (see Lemma 2.7 of \cite{clarkson2009numerical}). Thus, our alternative proof slightly increases the $\delta$ dependence stated in Lemma \ref{matapprox}.

\subsection{Column Sampling}
\emph{Feature selection} methods like column sampling are often preferred to \emph{feature extraction} methods like random projection or SVD reduction. Sampling produces an output matrix that is easier to interpret, indicating which original data dimensions are most `important'. Furthermore, the output sketch often maintains characteristics of the input data (e.g. sparsity) that may have substantial runtime and memory benefits when performing final data analysis.

The guarantees of family \emph{4} immediately imply that feature selection via column sampling suffices for obtaining a $(1+\epsilon)$ error projection-cost preserving sketch.
However, unlike the first three families, family \emph{4} is non-oblivious -- our column sampling probabilities and new column weights are computed using $\bv{B}$ and hence a low rank subspace $\bv Z$ satisfying the conditions of Lemma \ref{blocklem}. Specifically, the sampling probabilities in Lemma \ref{matapprox} are equivalent to the column norms of $\bv{Z}^\top$ added to a constant multiple of those of $\bv{A} - \bv{A}\bv{Z}\bv{Z}^\top$.  If $\bv{Z}$ is chosen to equal $\bv{V}_{2k}$ (as suggested for Lemma \ref{blocklem}), computing the subspace alone could be costly. So, we specifically structured Lemma \ref{blocklem} to allow for the use of an approximation to $\bv{V}_{2k}$. 
Additionally, we show that, once a suitable $\bv Z$ is identified, for instance using an approximate SVD algorithm, sampling probabilities can be approximated in nearly input-sparsity time, without having to explicitly compute $\bv B$. Formally, letting $nnz(\bv{A})$ be the number of non-zero entries in our data matrix $\bv{A}$,

\begin{lemma}
For any $\bv{A} \in \mathbb{R}^{n \times d}$, given an orthonormal basis $\bv{Z} \in \mathbb{R}^{d \times m}$ for a rank $m$ subspace of $\R^d$, for any $\delta$, there is an algorithm that can compute constant factor approximations of the column norms of $\bv{A} - \bv{A}\bv{Z}\bv{Z}^\top$ in time $O(nnz(\bv{A}) \log (d / \delta) + md \log (d / \delta))$ time, succeeding with probability $1 - \delta$.
\end{lemma}
Note that, as indicated in the statement of Lemma \ref{matapprox}, the sampling routine analyzed in \cite{stablechernoff} is robust to using norm overestimates. Scaling norms up by our constant approximation factor (to obtain strict overestimates) at most multiplies the number of columns sampled by a constant.

\begin{proof}
The approximation is obtained via a Johnson-Lindenstrauss transform.  To approximate the column norms of
$\bv{A} - \bv{A}\bv{Z}\bv{Z}^\top = \bv{A} (\bv{I} - \bv{Z}\bv{Z}^\top)$, we instead compute $\bs \Pi \bv{A} (\bv{I} - \bv{Z}\bv{Z}^\top)$, where $\bv \Pi$ is a Johnson-Lindenstrauss matrix with $O(\log (d/\delta) / \epsilon^2)$ rows drawn from, for example, family \emph{1} of Lemma \ref{matapprox}. By the standard Johnson-Lindenstrauss lemma \cite{Achlioptas:2003}, with probability at least $1-\delta$, every column norm will be preserved to within $1 \pm \epsilon$. We may fix $\epsilon = 1/2$.

Now, $\bv \Pi \bv{A} (\bv{I} - \bv{Z}\bv{Z}^\top)$ can be computed in steps.  First, compute $\bs \Pi \bv{A}$ by explicitly multiplying the matrices.  Since $\bv \Pi$ has $O(\log (d/\delta))$ rows, this takes time $O(nnz(\bv{A}) \log (d / \delta))$.  Next, multiply this matrix on the right by $\bv{Z}$ in time $O(md \log (d / \delta))$, giving $\bv \Pi \bv{A} \bv{Z}$, with $O(\log (d / \delta))$ rows and $m$ columns.  Next, multiply on the right by $\bv{Z}^\top$, giving $\bv \Pi \bv{A} \bv{Z} \bv{Z}^\top$, again in time $O(md \log(d / \delta))$.  Finally, subtracting from $\bv \Pi \bv{A}$ gives the desired matrix; the column norms can then be computed with a linear scan in time $O(d \log(d / \delta))$.
\end{proof}

Again, the sampling probabilities required for family \emph{4} are proportional to the sum of the column norms of $\bv{Z}^\top$ and a constant multiple of those of $\bv{A} - \bv{A}\bv{Z}\bv{Z}^\top$. Column norms of $\bv{Z}^\top$ take only linear time in the size of $\bv{Z}$ to compute, so the total runtime of computing sampling probabilities is $O(nnz(\bv{A}) \log (d / \delta) + md \log (d / \delta))$.

Finally, we address a further issue regarding the computation of $\bv{Z}$:  a generic approximate SVD algorithm may not satisfy the \emph{spectral norm} requirement on $\bv{A}-\bv{AZZ}^\top$ from Lemma \ref{blocklem}. Our analysis in Appendix \ref{spectral_version} can be used to obtain fast algorithms for approximate SVD that \emph{do} give the required spectral guarantee -- i.e. produce a $\bv{Z} \in \mathbb{R}^{d \times 2k}$ with $\norm{\bv{A}-\bv{A}\bv{Z}\bv{Z}^\top}_F^2 \leq \frac{2}{k}\norm{\bv{A}_{r \setminus k}}_F^2$. Nevertheless, it is possible to argue that even a conventional Frobenius norm error guarantee suffices.

The trick is to use a $\bv{Z}'$ in Lemma \ref{blocklem} that differs from the $\bv{Z}$ used to compute sampling probabilities. Specifically, we will choose a $\bv{Z}'$  that represents a potentially larger subspace.  Given a $\bv{Z}$ satisfying the Frobenius norm guarantee, consider the SVD of $\bv{A} - \bv{A}\bv{Z}\bv{Z}^\top$ and create $\bv{Z}'$ by appending to $\bv{Z}$ all singular directions with squared singular value $>\frac{2}{k} \norm{\bv{A}_{r \setminus k}}_F^2$.  This ensures that the spectral norm of the newly defined $\bv{A} - \bv{A}\bv{Z}'\bv{Z}'^\top$ is $\le \frac{2}{k} \norm{\bv{A}_{r \setminus k}}_F^2$. Additionally, we append at most $k$ rows to $\bv Z$. Since a standard approximate SVD can satisfy the Frobenius guarantee with a rank $k$ $\bv{Z}$, $\bv{Z}'$ has rank $\le 2k$, which is sufficient for Lemma \ref{blocklem}. 
 Furthermore, this procedure can only decrease column norms for the newly defined $\bv{B}'$: effectively, $\bv{B}'$ has all the same right singular vectors as $\bv{B}$, but with some squared singular values decreased from $> 2$ to 1. So, the column norms we compute will still be valid over estimates for the column norms of $\bv B$. 
Putting everything together gives:
\begin{theorem}
\label{sample}
For any $\bv{A} \in \mathbb{R}^{n \times d}$, given an orthonormal basis $\bv{Z} \in \mathbb{R}^{d \times k}$ satisfying $\norm{\bv{A}-\bv{A}\bv{Z}\bv{Z}^\top}_F^2 \le 2 \norm{\bv{A}_{r \setminus k}}_F^2$, for any $\epsilon < 1$ and $\delta$, there is an algorithm running in time $O(nnz(\bv{A}) \log (d / \delta) + kd \log (d / \delta))$ returning $\bv{\tilde A}$ containing $O(k \log(k / \delta) / \epsilon^2)$ reweighted columns of $\bv A$, such that, with probability at least $1-\delta$, $\bv{\tilde A}$ is a rank $k$ projection-cost preserving sketch for $\bv{A}$ (i.e. satisfies Definition \ref{def:2sidedsketch}) with error $\epsilon$.
\end{theorem}

It is worth noting the connection between our column sampling procedure and recent work on column based matrix reconstruction \cite{deshpande2006matrix,guruswami2012optimal,boutsidis2014near,boutsidis2014optimal}. Our result shows that it is possible to start with a constant factor approximate SVD of $\bv{A}$ and sample the columns of $\bv{A}$ by a combination of the row norms of $\bv{Z}$ and and the column norms of $\bv{A-AZZ}^T$. In other words, to sample by a combination of the \emph{leverage scores} with respect to $\bv{Z}$ and the \emph{residuals} after projecting the rows of $\bv{A}$ onto the subspace spanned by $\bv{Z}$.  In \cite{boutsidis2014optimal}, a very similar technique is used in Algorithm $1$. $\bv{A}$ is first sampled according to the leverage scores with respect to $\bv{Z}$. Then, in the process referred to as \emph{adaptive sampling}, $\bv{A}$ is sampled by the column norms of the residuals after $\bv{A}$ is projected to the columns selected in the first round (see Section 3.4.3 of \cite{boutsidis2014optimal} for details on the adaptive sampling procedure). Intuitively, our single-shot procedure avoids this adaptive step by incorporating residual probabilities into the initial sampling probabilities. 

Additionally, note that our procedure recovers a projection-cost preserving sketch with $\tilde O(k/\epsilon^2)$ columns. In other words, if we compute the top $k$ singular vectors of our sketch, projecting to these vectors will give a $(1+\epsilon)$ approximate low rank approximation to $\bv{A}$.  In \cite{boutsidis2014optimal}, the $1/\epsilon$ dependence is linear, rather than quadratic, but the selected columns satisfy a weaker notion: that there exists some good $k$-rank approximation falling within the span of the selected columns.

\subsection{Deterministic Column Selection}
Finally, family \emph{5} gives an algorithm for feature selection that produces a $(1+\epsilon)$ projection-cost preserving sketch with just $O(k/\epsilon^2)$ columns.
The \emph{BSS Algorithm} is a deterministic procedure introduced in \cite{batson2012twice} for selecting rows from a matrix $\bv{M}$ using a selection matrix $\bv{R}$ so that $\norm{\bv{M}^\top \bv{R}^\top \bv R \bv M - \bv{M}^\top \bv M}_2 \leq \epsilon$. The algorithm is slow -- it runs in $poly(n,q,\epsilon)$ time for an $\bv{M}$ with $n$ columns and $q$ rows. However, the procedure can be advantageous over sampling methods like family \emph{4} because it reduces a rank $k$ matrix to $O(k)$ dimensions instead of $O(k\log k)$. \cite{jelaniDiscussion} extends this result to matrices with stable rank $\leq k$.

Furthermore, it is possible to substantially reduce runtime of the procedure in practice. $\bv{A}$ can first be sampled down to $O(k \log k / \epsilon^2)$ columns using Theorem \ref{sample} to produce $\bv{\ol A}$. Additionally, as for family \emph{4}, instead of fully computing $\bv{\ol B}$, we can compute $\bv{\Pi \ol B}$ where $\bv \Pi$ is a sparse subspace embedding (for example from family \emph{2}).  $\bv{\Pi \ol B}$ will have dimension just $O((k\log k)^2/\epsilon^6) \times O(k\log k /\epsilon^2)$. As $\bv{\Pi}$ will preserve the spectral norm of $\bv{\ol B}$, it is clear that the column subset chosen for $\bv{\Pi \ol B}$ will also be a valid subset for $\bv{\ol B}$. Overall this strategy gives:
\begin{theorem}
\label{bss_theorem}
For any $\bv{A} \in \mathbb{R}^{n \times d}$ and any $\epsilon < 1$, $\delta > 0$, there is an algorithm running in time $O(nnz(\bv{A}) \log (d / \delta) + poly(k, \epsilon, \log(1/\delta))d)$ which returns $\bv{\tilde A}$ containing $O(k/ \epsilon^2)$ reweighted columns of $\bv A$, such that, with probability at least $1-\delta$, $\bv{\tilde A}$ is a rank $k$ projection-cost preserving sketch for $\bv{A}$ (i.e. satisfies Definition \ref{def:2sidedsketch}) with error $\epsilon$.
\end{theorem}
\section{Non-Oblivious Random Projection}\label{squish}

In this section, we show how to obtain projection-cost preserving sketches using a non-oblivious random projection technique that is standard for approximate SVD algorithms \cite{sarlos2006improved,Clarkson:2013}. To obtain a sketch of $\bv{A}$, we first multiply on the left by a Johnson-Lindenstrauss matrix with $O(k/\epsilon)$ rows. We then project the rows of $\bv{A}$ onto the row span of this much shorter matrix to obtain $\bv{\tilde A}$. In this way, we have projected $\bv{A}$ to a random subspace, albeit one that depends on the rows of $\bv{A}$ (i.e. non-obliviously chosen). This method gives an improved $\epsilon$ dependence over the oblivious approach of multiplying $\bv{A}$ on the right by a single Johnson-Lindenstrauss matrix (Theorem \ref{rp_theorem}). 
Specifically, we show:

\begin{theorem}\label{non_oblivious_full_theorem}
For $0 \le \epsilon < 1$, let $\bv{R}$ be drawn from one of the first three Johnson-Lindenstrauss distributions of Lemma \ref{matapprox} with $\epsilon' = O(1)$ and $k' = O(k/\epsilon)$. Then, for any $\bv{A} \in \mathbb{R}^{n \times d}$,  let $\bv{\ol A} = \bv{R A}$ and let $\bv{Z}$ be a matrix whose columns form an orthonormal basis for the rowspan of $\bv{\ol A}$. With probability $1-\delta$, $\bv{\tilde A} = \bv{A}\bv{Z}$ is a projection-cost preserving sketch for $\bv{A}$ satisfying the conditions of Definition \ref{def:1sidedsketch} with error $\epsilon$.
\end{theorem}

As an example, if $\bv R$ is a dense Johnson-Lindenstrauss matrix (family \emph{1} in Lemma \ref{matapprox}), it will reduce $\bv{A}$ to $O(\frac{k' + \log(1/\delta)}{\epsilon'^2}) = O(k/\epsilon + \log(1/\delta))$ rows and thus $\bv{A Z}$ will have dimension $O(k/\epsilon + \log(1/\delta))$.

As usual, we actually show that $\bv{A}\bv{Z}\bv{Z}^\top$ is a projection-cost preserving sketch and note that $\bv{AZ}$ is as well since it is simply a rotation. Our proof requires two steps. In Theorem \ref{approx_svd}, we showed that any rank $\lceil k/\epsilon \rceil$  approximation for $\bv{A}$ with Frobenius norm cost at most $(1+\epsilon)$ from optimal yields a projection-cost preserving sketch. Here we start by showing that any low rank approximation with small $\emph{spectral norm}$ cost also suffices as a projection-cost preserving sketch. We then show that non-oblivious random projection to $O(k/\epsilon)$ dimensions gives such a low rank approximation, completing the proof. The spectral norm low rank approximation result follows:

\begin{lemma}\label{spectral_approx_svd}
For any $\bv{A} \in \mathbb{R}^{n \times d}$ and any orthonormal matrix $\bv{Z} \in \mathbb{R}^{d \times m}$ satisfying $\norm{\bv{A}-\bv{A}\bv{Z}\bv{Z}^\top}_2^2 \le \frac{\epsilon}{k} \norm{\bv{A}_{r \setminus k}}_F^2$, the sketch $\bv{\tilde A} = \bv{A}\bv{Z}\bv{Z}^\top$ satisfies the conditions of Definition \ref{def:1sidedsketch}. Specifically, for all rank $k$ orthogonal projections $\bv{P}$,
\begin{align*}
\norm{\bv{A} - \bv{P}\bv{A}}_F^2 \le \norm{\bv{\tilde A-P\tilde A}}_F^2 + c \le (1+\epsilon)  \norm{\bv{A} - \bv{P}\bv{A}}_F^2.
\end{align*}
\end{lemma}
\begin{proof}
As in the original approximate SVD proof (Theorem \ref{approx_svd}), we set $\bv{E} = - (\bv{A} - \bv{A}\bv{Z}\bv{Z}^\top)(\bv{A} - \bv{A}\bv{Z}\bv{Z}^\top)^\top$. $\bv{\tilde C} = \bv{C} + \bv{E}$, $\bv{E}$ is symmetric, and $\bv{E} \preceq \bv{0}$. Furthermore, by our spectral norm approximation bound,
\begin{align*}
\sum_{i=1}^k |\lambda_i(\bv{E})| \le k  \norm{\bv{A} - \bv{A}\bv{Z}\bv{Z}^\top}_2^2 \le \epsilon \norm{\bv{A}_{r \setminus k}}_F^2.
\end{align*}
The result then follows directly from Lemma \ref{error_improved}.
\end{proof}

Next we show that the non-oblivious random projection technique described satisfies the spectral norm condition required for Lemma \ref{spectral_approx_svd}. Combining these results gives us Theorem \ref{non_oblivious_full_theorem}.

\begin{lemma} 
\label{spectral_norm_nonoblivious}
For $0 \le \epsilon < 1$, let $\bv{R}$ be drawn from one of the first three distributions of Lemma \ref{matapprox} with $\epsilon' = O(1)$ and $k' = \lceil k/\epsilon \rceil + k -1$. Then, for any $\bv{A} \in \mathbb{R}^{n \times d}$,  let $\bv{\ol A} = \bv{R A}$ and let $\bv{Z}$ be a matrix whose columns form an orthonormal basis for the rowspan of $\bv{\ol A}$. Then, 
with probability $1-\delta$,
\begin{align}\label{small_spectral_vs_frob}
\norm{\bv{A}-\bv{A}\bv{Z}\bv{Z}^\top}_2^2 \le O\left( \frac{\epsilon}{k} \right ) \norm{\bv{A}_{r \setminus k}}_F^2.
\end{align}
\end{lemma}
\begin{proof}
To prove this Lemma, we actually consider an alternative projection technique: multiply $\bv{A}$ on the left by $\bv R$ to obtain $\bv{\ol A}$, find its best rank $k'$ approximation $\ol{\bv{A}}_{k'}$, then project the rows of $\bv{A}$ onto the rows of $\ol{\bv{A}}_{k'}$. Letting $\bv{Z}'$ be a matrix whose columns are an orthonormal basis for the rows of $\ol{\bv{A}}_{k'}$, it is clear that
\begin{align}\label{only_better}
\norm{\bv{A}-\bv{A}\bv{Z}\bv{Z}^\top}_2^2 \le \norm{\bv{A}-\bv{A}\bv{Z'}\bv{Z'}^\top}_2^2.
\end{align}
$\ol{\bv{A}}_{k'}$'s rows fall within the row span of $\bv{\ol A}$, so the result of projecting to the orthogonal complement of $\bv{\ol A}$'s rows is unchanged if we first project to the orthogonal complement of $\ol{\bv{ A}}_{ k'}$'s rows. Then, since projection can only decrease spectral norm,
\begin{align*}
\norm{\bv{A}(\bv{I} -\bv{Z}\bv{Z}^\top)}_2^2 = \norm{ \bv{A}(\bv{I} -\bv{Z'}\bv{Z'}^\top)( \bv{I} -\bv{Z}\bv{Z}^\top)}_2^2 \le \norm{\bv{A} ( \bv{I} -\bv{Z'}\bv{Z'}^\top)}_2^2,
\end{align*}
giving Equation \eqref{only_better}.

So we just need to show that $\norm{\bv{A}-\bv{A}\bv{Z'}\bv{Z'}^\top}_2^2 \le  \frac{\epsilon}{k} \norm{\bv{A}_{r \setminus k}}_F^2$. Note that, since $k' = \lceil k/\epsilon \rceil + k -1$,
\begin{align*}
\norm{\bv{A}_{r \setminus k'}}_2^2 = \sigma_{k'+1}^2(\bv{A}) \le \frac{1}{k} \sum_{i=k'+2-k}^{k'+1} \sigma_i^2(\bv{A}) \le \frac{\epsilon}{k} \sum_{i=k+1}^{k'+2-k} \sigma_i^2(\bv{A}) \le \frac{\epsilon}{k} \norm{\bv{A}_{r \setminus k}}_F^2.
\end{align*}
Additionally, $\norm{\bv{A}_{r\setminus k'}}_F^2 \le \norm{\bv{A}_{r\setminus k}}_F^2$ and $\frac{1}{k'} \le \frac{k}{\epsilon}$. So to prove \eqref{small_spectral_vs_frob} it suffices to show:
\begin{align*}
\norm{\bv{A}-\bv{A}\bv{Z'}\bv{Z'}^\top}_2^2 \le O(1) \left(\norm{\bv{A}_{r\setminus k'}}_2^2 + \frac{1}{k'}\norm{\bv{A}_{r\setminus k'}}_F^2 \right).
\end{align*}
In fact, this is just a just an approximate SVD with a spectral norm guarantee, similar to what we have already shown for the Frobenius norm! Specifically, $\bv{Z'}$ is an approximate $k'$ SVD, computed using a projection-cost preserving sketch as given in Theorem \ref{rp_theorem} with $\epsilon' = O(1)$. 
Here, rather than a multiplicative error on the Frobenius norm, we require a multiplicative error on the spectral norm, plus a small additive Frobenius norm error.
Extending our Frobenius norm approximation guarantees to give this requirement is
straightforward but tedious. The result is included in Appendix \ref{spectral_version}, giving us Lemma \ref{spectral_norm_nonoblivious} and thus Theorem \ref{non_oblivious_full_theorem}. We also note that a sufficient bound is given in Theorem 10.8 of \cite{Halko:2011}, however we include an independent proof for completeness and to illustrate the application of our techniques to spectral norm approximation guarantees.

%

\end{proof}
\section{Constant Factor Approximation with $O(\log k)$ Dimensions}\label{logk}

In this section we show that randomly projecting $\bv{A}$ to just $O(\log k/\epsilon^2)$ dimensions using a Johnson-Lindenstrauss matrix is sufficient for approximating $k$-means up to a factor of $(9+\epsilon)$. To the best of our knowledge, this is the first result achieving a constant factor approximation using a sketch with data dimension independent of the input size ($n$ and $d$) and sublinear in $k$. This result opens up the interesting question of whether is is possible to achieve a $(1+\epsilon)$ relative error approximation to $k$-means using just $O(\log k)$ rather than $O(k)$ dimensions. Specifically, we show:

\begin{theorem}\label{logk_approx} For any $\bv{A} \in \mathbb{R}^{n \times d}$, any $0 \le \epsilon < 1$, and $\bv{R} \in \mathbb{R}^{O (\frac{\log (k/\delta)}{\epsilon^2} ) \times d}$ drawn from a Johnson-Lindenstrauss distribution,  let $\bv{\tilde A} = \bv{AR^\top}$. Let $S$ be the set of all $k$-cluster projection matrices, let $\bv{P^*} = \argmin_{\bv{P} \in S} \norm{\bv{A} - \bv{P}\bv{ A}}_F^2$, and let $\bv{\tilde P}^* = \argmin_{\bv{P} \in S} \norm{\bv{\tilde A} - \bv{P}\bv{\tilde A}}_F^2$. With probability $1-\delta$, for any $\gamma \ge 1$, and $\bv{\tilde P} \in S$,
if $\norm{\bv{\tilde A} - \bv{\tilde P} \bv{\tilde A}}_F^2 \le \gamma \norm{\bv{\tilde A} - \bv{\tilde P}^* \bv{\tilde A}}_F^2$:
\begin{align*}
\norm{\bv{ A} - \bv{\tilde P} \bv{ A}}_F^2 \le (9+\epsilon) \cdot \gamma \norm{\bv{A}-\bv{P^* A}}_F^2.
\end{align*}
\end{theorem}
In other words, if $\bv{\tilde P}$ is a cluster indicator matrix (see Section \ref{subsec:k_means_constrained}) for an approximately optimal clustering of $\bv{\tilde A}$, then the clustering is also within a constant factor of optimal for $\bv{A}$. Note that there are a variety of distributions that are sufficient for choosing $\bv{R}$. For example, we may use the dense Rademacher matrix distribution of family \emph{1} of Lemma \ref{matapprox}, or a sparse family such as those given in \cite{kane2014sparser}.

To achieve the $O(\log k/\epsilon^2)$ bound, we must focus specifically on $k$-means clustering -- it is clear that projecting to $< k$ dimensions is insufficient for solving general constrained $k$-rank approximation as $\bv{\tilde A}$ will not even have rank $k$. Additionally, random projection is the only sketching technique of those studied that can work when $\bv{\tilde A}$ has fewer than O($k$) columns. 
Consider clustering the rows of the $n \times n$ identity into $n$ clusters, achieving cost $0$. An SVD projecting to less than $k=n-1$ dimensions or column selection technique taking less than $k=n-1$ columns will leave at least two rows in $\bv{\tilde A}$ with all zeros. These rows may be clustered together when optimizing the $k$-means objective for $\bv{\tilde A}$, giving a clustering with cost $> 0$ for $\bv{A}$ and hence failing to achieve multiplicative error.

\begin{proof}
As mentioned in Section \ref{subsec:our_approach}, the main idea is to analyze an $O(\log k/\epsilon^2)$ dimension random projection by splitting $\bv{A}$ in a substantially different way than we did in the analysis of other sketches. Specifically, 
%
%
we split it according to its optimal $k$ clustering and the remainder matrix:
\begin{align*}
\bv{A} = \bv{P}^* \bv{A} + (\bv{I-P}^*) \bv{A}.
\end{align*}
For conciseness, write $\bv{B} =  \bv{P}^* \bv{A}$ and $\bv{\ol B} = (\bv{I-P}^*) \bv{A}$. So we have $\bv{A} = \bv{B} + \bv{\ol B}$ and $\bv{\tilde A} =  \bv{B R}^\top + \bv{\ol B}\bv{ R}^\top$.

By the triangle inequality and the fact that projection can only decrease Frobenius norm:
\begin{align}
\label{first_triangle}
\norm{\bv{ A} - \bv{\tilde P} \bv{ A}}_F \le \norm{\bv{B} - \bv{\tilde P} \bv{B}}_F + \norm{\bv{ \ol B} - \bv{\tilde P} \bv{\ol B}}_F \le \norm{\bv{B} - \bv{\tilde P} \bv{B}}_F + \norm{\bv{\ol B}}_F.
\end{align}

Next note that $\bv{B}$ is simply $\bv{A}$ with every row replaced by its cluster center (in the optimal clustering of $\bv{A}$). So $\bv{B}$ has just $k$ distinct rows. Multiplying by a Johnson-Lindenstauss matrix with $O(\log (k/\delta)/\epsilon^2)$ columns will preserve the squared distances between all of these $k$ points with high probability. It is not difficult to see that preserving distances is sufficient to preserve the cost of any clustering of $\bv{B}$ since we can rewrite the $k$-means objection function as a linear function of squared distances alone:
\begin{align*}
\norm{\bv{B} - \bv{X}_C\bv{X}_C^\top \bv{B}}_F^2 = \sum_{j=1}^n \|\bv{b}_{j} - \bs{\mu}_{C(j)}\|_2^2 = \sum_{i = 1}^k \frac{1}{|C_i|} \sum_{\substack{\bv{b}_j,\bv{b}_k \in C_i \\ j \neq k}} \norm {\bv{b}_j - \bv{b}_k}_2^2.
\end{align*}
So, $\norm{\bv{B} - \bv{\tilde P} \bv{B}}_F^2 \le (1+\epsilon)\norm{\bv{BR}^\top - \bv{\tilde P} \bv{BR}^\top}_F^2$. Combining with \eqref{first_triangle} and noting that square rooting can only reduce multiplicative error, we have:
\begin{align*}
\norm{\bv{ A} - \bv{\tilde P} \bv{ A}}_F \le (1+\epsilon)\norm{\bv{B}\bv{R}^\top - \bv{\tilde P} \bv{BR}^\top}_F + \norm{\bv{\ol B}}_F.
\end{align*}
Rewriting $\bv{B}\bv{R}^\top = \bv{\tilde A} - \bv{\ol B}\bv{ R}^\top$ and again applying triangle inequality and the fact the projection can only decrease Frobenius norm, we have:
\begin{align*}
\norm{\bv{ A} - \bv{\tilde P} \bv{ A}}_F &\le (1+\epsilon)\norm{(\bv{\tilde A} - \bv{\ol B}\bv{ R}^\top) - \bv{\tilde P} (\bv{\tilde A} - \bv{\ol B}\bv{ R}^\top) }_F + \norm{\bv{\ol B}}_F\\
&\le (1+\epsilon)\norm{\bv{\tilde A} - \bv{\tilde P} \bv{\tilde A}}_F +  (1+\epsilon)\norm{(\bv{I}-\bv{\tilde P})\bv{\ol B}\bv{R}^\top}_F +  \norm{\bv{\ol B}}_F\\
&\le (1+\epsilon)\norm{\bv{\tilde A} - \bv{\tilde P} \bv{\tilde A}}_F +  (1+\epsilon)\norm{\bv{\ol B}\bv{ R}^\top}_F +  \norm{\bv{\ol B}}_F
\end{align*}

As discussed in Section \ref{spectral_norm_proofs}, multiplying by a Johnson-Lindenstrauss matrix with at least $O(\log(1/\delta)/\epsilon^2)$ columns will preserve the Frobenius norm of any fixed matrix up to $\epsilon$ error so $\norm{\bv{\ol B}\bv{ R}^\top}_F \le (1+\epsilon) \norm{\bv{\ol B}}_F$. Using this and the fact that $\norm{\bv{\tilde A} - \bv{\tilde P} \bv{\tilde A}}_F^2 \le \gamma \norm{\bv{\tilde A} - \bv{\tilde P}^* \bv{\tilde A}}_F^2 \le  \gamma \norm{\bv{\tilde A} - \bv{P}^* \bv{\tilde A}}_F^2$ we have:
\begin{align*}
\norm{\bv{ A} - \bv{\tilde P} \bv{ A}}_F &\le (1+\epsilon)\sqrt{\gamma} \norm{\bv{\tilde A} - \bv{P}^* \bv{\tilde A}}_F + (2+3\epsilon) \norm{\bv{\ol B}}_F.
\end{align*}

Finally, we note that $\bv{\ol B} = \bv{A}-\bv{P^* A}$ and again apply the fact that multiplying by $\bv{R}^\top$ preserves the Frobenius norm of any fixed matrix with high probability. So, $\norm{\bv{\tilde A} - \bv{P}^* \bv{\tilde A}}_F \le (1+\epsilon)\norm{\bv{A} - \bv{P}^* \bv{A}}_F$ and thus:
\begin{align*}
\norm{\bv{ A} - \bv{\tilde P} \bv{ A}}_F  \le (3+6\epsilon) \sqrt{\gamma}\norm{\bv{A} - \bv{P}^* \bv{A}}_F.
\end{align*}
Squaring and adjusting $\epsilon$ by a constant factor gives the desired result.
\end{proof}
\section{Applications to Streaming and Distributed Algorithms}\label{applications}

As mentioned, there has been an enormous amount of work on exact and approximate $k$-means clustering algorithms \cite{Inaba:1994:AWV:177424.178042,kanungo2002local,1366265,Arthur2007,Har-Peled:2007:SCK:1229133.1229135}. While surveying all relevant work is beyond the scope of this paper, applying our dimensionality reduction results black box gives immediate improvements to existing algorithms with runtime dependence on dimension.

Our results also have a variety of applications to distributed and streaming algorithms. The size of coresets for $k$-means clustering typically depend on data dimension, so our relative error sketches with just $\lceil k/\epsilon \rceil$ dimensions and constant error sketches with $O(\log k)$ dimensions give the smallest known constructions. See \cite{har2004coresets,Har-Peled:2007:SCK:1229133.1229135,balcan2013distributed,feldman2013turning} for more information on coresets and their use in approximation algorithms as well as distributed and streaming computation. Aside from these immediate results, we briefly describe two example applications of our work.

\subsection{Streaming Low Rank Approximation}

For any matrix $\bv{A} \in \mathbb{R}^{n \times d}$, consider the problem of finding a basis for an approximately optimal $k$-rank subspace to project the rows of $\bv{A}$ onto -- i.e. computing an approximate SVD like the one required for Theorem \ref{approx_svd}.  That is, we wish to find $\bv{Z} \in \mathbb{R}^{d \times k}$ such that
\begin{align*}
\norm{\bv{A} - \bv{A} \bv{ZZ}^\top }_F^2 \le (1+\epsilon) \norm{\bv{A}_{r \setminus k}}_F^2
\end{align*} 

Building on the work of \cite{liberty2013simple}, \cite{ghashami2014relative} gives a deterministic algorithm for this problem using $O(dk/\epsilon)$ words of space in the row-wise streaming model, when the matrix $\bv{A}$ is presented to and processed by a server one row at a time. \cite{woodruffNew} gives a nearly matching lower bound, showing that $O(dk/\epsilon)$ bits of space is necessary for solving the problem, even using a randomized algorithm with constant failure probability. 

Theorem \ref{rp_theorem} applied to unconstrained $k$-rank approximation allows this problem to be solved using $\tilde O(dk/\epsilon^2)$ words and $\tilde O(\log k \log n)$ bits of space in the general turnstile streaming model where arbitrary additive updates to entries in $\bv{A}$ are presented in a stream. Word size is typically assumed to be $O(\log d \log n)$ bits, giving us an $\tilde O(dk/\epsilon^2)$ word space bound.  Here $\tilde O(\cdot )$ hides $\log$ factors in $k$ and the failure probability.

We simply sketch $\bv{A}$ by multiplying on the left by an $\tilde O(k/\epsilon^2) \times n$ matrix drawn from family \emph{3} of Lemma \ref{matapprox}, which only takes $\tilde O(\log k \log n)$ bits to specify. We then obtain $\bv Z$ by computing the top $k$ singular vectors of the sketch. 
This approach gives the best known bound in the turnstile streaming model using only a single pass over $\bv{A}$, nearly matching the $O(dk/\epsilon)$ lower bound given for the more restrictive row-wise streaming model. Previously approximate SVD algorithms \cite{sarlos2006improved,Clarkson:2013} rely on non-oblivious random projection, so could not give such a result.

\subsection{Distributed $k$-means clustering}

In \cite{balcan2013distributed}, the authors give a distributed $k$-means clustering algorithm for the setting where the rows of the data matrix $\bv{A}$ are arbitrarily partitioned across $s$ servers. Assuming that all servers are able to communicate with a central coordinator in one hop, their algorithm requires total communication $\tilde O(kd+ sk)$ (hiding dependence on error $\epsilon$ and failure probability $\delta$). A recent line of work \cite{liang2013distributed,kannan2014principal,balcan2014improved} seeks to improve the communication complexity of this algorithm by applying the SVD based dimensionality reduction result of \cite{feldman2013turning}. The basic idea is to apply a distributed SVD algorithm (also referred to as distributed PCA) to compute the top right singular vectors of $\bv{A}$. Each server can then locally project its data rows onto these singular vectors before applying the clustering algorithm from  \cite{balcan2013distributed}, which will use $\tilde O(kd'+ sk)$ communication, where $d'$ is the dimension we reduce down to. 

By noting that we can set $d'$ to $\lceil k/\epsilon \rceil$ instead of $O(k/\epsilon^2)$, we can further improve on the $k$-means communication complexity gains in this prior work. 
Additionally, our oblivious random projection result (Theorem \ref{rp_theorem}) can be used to avoid the distributed PCA preprocessing step entirely. Inherently, PCA requires $O(sdk)$ total communication -- see Theorem 1.2 of \cite{kannan2014principal} for a lower bound. Intuitively, the cost stems from the fact that $O(k)$ singular vectors, each in $\mathbb{R}^d$, must be shared amongst the $s$ servers. Using Theorem \ref{rp_theorem}, a central coordinator can instead send out bits specifying a single Johnson-Lindenstrauss matrix to the $s$ servers. Each server can then project its data down to just $\tilde O(k/\epsilon^2)$ dimensions and proceed to run the $k$-means clustering algorithm of \cite{balcan2013distributed}. They could also further reduce down to $\lceil k/\epsilon \rceil$ dimensions using a distributed PCA algorithm or to $O(k/\epsilon)$ dimensions using our non-oblivious random projection technique. Formalizing one possible strategy, we give the first result with communication only logarithmic in the input dimension $d$.

\begin{corollary} Given a matrix $\bv{A} \in \mathbb{R}^{n \times d}$ whose rows are partitioned across $s$ servers that are all connected to a single coordinator server, along with a centralized $\gamma$-approximate algorithm for $k$-means clustering, there is a distributed algorithm computing a $(1+\epsilon)\gamma$-approximation to the optimal clustering that succeeds with probability at least $1-\delta$ and communicates just $\tilde O(s\log d  \log k)$ bits, $\tilde O\left(\frac{sk}{\epsilon}\right)$ vectors in $\mathbb{R}^{\tilde O(k/\epsilon^2)}$, and $O\left (\frac{1}{\epsilon^4} \left (\frac{k^2}{\epsilon} + \log 1/\delta \right ) + sk\log \frac{sk}{\delta} \right )$ vectors in $\mathbb{R}^{O(k/\epsilon)}$.
\end{corollary}
\begin{proof}
Here $\tilde O(\cdot)$ hides log factors in the failure probability $\delta$.
For the initial reduction to $\tilde O\left(k/\epsilon^2\right)$ dimensions, we can choose a matrix from family \emph{3} of Lemma \ref{matapprox} that can be specified with $\tilde O(\log d  \log k)$ bits, which must be communicated to all $s$ servers.

We can then use Theorem \ref{non_oblivious_full_theorem} to further reduce to $O(k/\epsilon)$ dimensions. Note that the first three families of Lemma \ref{matapprox} all have independent columns. So, in order to compute $\bv{R A}$ where $\bv{R} \in \mathbb{R}^{\tilde O(k/\epsilon) \times n}$ is drawn from one of these families, each server can simply independently choose $\bv{R}_i \in \mathbb{R}^{\tilde O(k/\epsilon) \times n_i}$ from the same distribution, compute $\bv{R}_i \bv{A}_i$, and send it to the central server. Here $\bv{A}_i$ is the set of rows held by server $i$ and $n_i$ is the number of rows in $\bv{A}_i$. The central server can then just compute $\bv{RA} = \sum_{i=1}^s \bv{R}_i \bv{A}_i$, and send back an orthonormal basis for the rows of $\bv{RA}$ to the servers. To further reduce dimension from $\tilde O\left(k/\epsilon\right)$ to $O(k/\epsilon)$, and to improve constant factors, the central server can actually just return an orthonormal basis for the best rank $O(k/\epsilon)$ approximation of $\bv{RA}$, as described in the proof of Lemma \ref{spectral_norm_nonoblivious}. Each server can then independently project their rows to this basis. The total communication of this procedure is $\tilde O\left(\frac{sk}{\epsilon}\right)$ vectors in $\mathbb{R}^{\tilde O(k/\epsilon^2)}$.

Finally, applying Theorem 3 of \cite{balcan2013distributed} with $h =1$ and $d = O(k/\epsilon)$ and adjusting $\epsilon$ by a constant factor gives a communication cost of $O\left (\frac{1}{\epsilon^4} \left (\frac{k^2}{\epsilon} + \log 1/\delta \right ) + sk\log \frac{sk}{\delta} \right )$ vectors in $\mathbb{R}^{O(k/\epsilon)}$ for solving the final clustering problem to within $(1+\epsilon)\gamma$ error.
\end{proof}

\section{Open Questions}

As mentioned, whether it is possible to improve on our $(9+\epsilon)$ $k$-means approximation guarantee for random projection to $O(\log k/\epsilon^2)$ dimensions is an intriguing open question. 

We are also interested in whether our column sampling results can be used to develop fast low rank approximation algorithms based on sampling. Theorem \ref{sample} requires a constant factor approximate SVD and returns a sketch from which one can compute a $(1+\epsilon)$ factor approximate SVD. In other words, it gives a method for refining a coarse approximate SVD to a relative error one. Is it possible to start with an even coarser approximate SVD or set of sampling probabilities and use this refinement procedure to iteratively obtain better sampling probabilities and eventually a relative error approximate SVD? Such an algorithm would only ever require computing exact SVDs on small column samples of the original matrix, possibly leading to advantages over Johnson-Lindenstrauss type methods if the original matrix, and hence each sample, is sparse or structured. Iterative algorithms of this form have been developed for approximate regression \cite{peng_iterative,cohen2015uniform}. Extending these results to low rank approximation is an interesting open question.

\section{Acknowledgements} We thank Piotr Indyk, Jonathan Kelner, and Aaron Sidford for helpful conversations and Ludwig Schmidt and Yin Tat Lee for valuable edits and comments on the writeup. We also thank Jelani Nelson and David Woodruff for extensive discussion and guidance. This work was partially supported by National Science Foundation awards,
CCF-1111109,
CCF-0937274,  CCF-0939370, CCF-1217506, and IIS-0835652, NSF Graduate Research Fellowship Grant No. 1122374, 
AFOSR grant FA9550-13-1-0042, and DARPA grant FA8650-11-C-7192.

\bibliography{kmeans}{}
\bibliographystyle{alpha}

\appendix
\section{Matching Lower Bound for SVD Based Reduction}\label{lower_bound}

We show that to approximate the $k$-means objective function to within $(1+\epsilon)$ using the singular value decomposition, it is necessary to use at least the best rank $\lceil k/\epsilon \rceil$ approximation to $\mathbf A$. This lower bound clearly also implies that $\lceil k/\epsilon \rceil$ is necessary for all constrained low rank approximation problems, of which $k$-means is a specific instance. Technically, we prove:
\begin{theorem}\label{lb_theorem}
For any $\lambda<1$ and $\epsilon>0$, there exist $n,d,k$ and $\mathbf A\in\R^{n\times d}$ such that, choosing $m=\lceil k/\epsilon \rceil$ and letting $\bv{\tilde P}$ be the $k$ cluster projection matrix minimizing $\norm{\bv{A}_m - \bv{P} \bv{A}_m}_F^2$ and $\bv{P}^*$ be the cluster projection minimizing $\norm{\bv{A} - \bv{P} \bv{A}}_F^2$ then:
\begin{align*}
\norm{\bv{A}- \bv{\tilde P} \bv{A}}_F^2 > (1+\epsilon) \lambda \norm{\bv{A} - \bv{P}^* \bv{A}}_F^2.
\end{align*}
\end{theorem}
First, we describe the data set $\bv A$ which proves the lower bound. We set $d = \lceil k/\epsilon\rceil + k -1$. In $k-1$ of the dimensions, we place a simplex. Orthogonal to this simplex in the other $\lceil k/\epsilon\rceil=m$ dimensions, we place a large number of random Gaussian vectors, forming a `cloud' of points. (Note: From now on we will drop the ceiling notation and will fix $k$ and $\epsilon$ so $k/\epsilon$ is an integer.) The Gaussian vectors will cluster naturally with only one center, so as one possible $k$ clustering, we can simply place $k-1$ of the cluster centers on the simplex and one of the centers at the centroid of the Gaussians, which will be near the origin. However, we will choose our points such that the largest singular directions will all be in the Gaussian cloud, so $\bv A_m$ will collapse the simplex to the origin and therefore any optimal clustering will keep the points of the simplex in a single cluster. This frees more clusters to use on the Gaussian cloud, but we will show that clustering the Gaussian cloud with $k$ clusters rather than $1$ cluster will not significantly reduce its clustering cost, so the increased cost due to the simplex will induce an error of almost an $\epsilon$ fraction of the optimal cost, giving us the lower bound.

Formally, choose $n=n'+k-1$ large (we will say precisely how large later). Define
\[\bv A=\begin{pmatrix}\lambda'\bv I_{k-1}&0\\0&\bv G\end{pmatrix},\] where $\lambda'$ will be a constant slightly smaller than 1, and $\bv{G} \in\R^{n'\times m}$ is a random Gaussian matrix with independent $\mathcal N(0,1/n')$ entries. The first $k-1$ rows form a simplex, and the remaining rows form the `Gaussian cloud.'

We need two properties of $\bv G$, which are given by the following Lemmas.
\begin{lemma}\label{gaussian_singular}
With high probability, the smallest singular value of $\bv G$ is at least $1-2\sqrt{m/n'}$.
\end{lemma}
\begin{proof}As a random Gaussian matrix, Theorem II.13 in \cite{davidson2001local} shows that the expected value of the smallest singular value of $\bv G$ is $1-\sqrt{m/n'}$. Rudelson and Vershynin \cite{rudelson2009smallest} cite this result and comment that it can be turned into a concentration inequality of the following form:
\[\Pr[\sigma_m(G)\le1-\sqrt{m/n'}-t/\sqrt{n'}]\le e^{-t^2/2}.\]
Taking $t=\sqrt m$ yields the result with probability exponentially close to 1.\end{proof}
Therefore, set $\lambda'=1-2\sqrt{\frac k{n'\epsilon}}$. This lemma guarantees that the $m$ largest singular values are associated with the Gaussian cloud, and therefore, the best rank $m$ approximation to $\bv A$ is $\bv A_m=\begin{pmatrix}0&0\\0&\bv G\end{pmatrix}$.

\begin{lemma}\label{unclusterable_cloud}
With high probability, for large enough $k$ and $n$, the optimal cost of clustering $\bv G$ with $k$ clusters is a $\lambda$-fraction of the optimal cost of clustering $\bv G$ with one cluster.
\end{lemma}
The idea of Lemma \ref{unclusterable_cloud} is simple: a cloud of random Gaussians is very naturally clustered with one cluster. However its proof is somewhat involved, so we first use Lemmas \ref{gaussian_singular} and \ref{unclusterable_cloud} to prove Theorem \ref{lb_theorem} and then prove Lemma \ref{unclusterable_cloud} at the end of the section.

\begin{proof}[Proof of Theorem \ref{lb_theorem}]

First, we analyze the clustering projection matrix
\[\bv P=\begin{pmatrix}\bv{I}_{k-1}&0\\0&\frac1{n'}\bv J_{n'}\end{pmatrix},\]
where $\bv J_{n'}$ is the all-ones matrix. This puts the whole Gaussian cloud in one cluster and the vertices of the simplex in their own clusters. We have
\[\norm{\bv A-\bv P\bv A}_F^2=\norm{\bv G-\frac1{n'}\bv J_{n'}\bv G}_F^2\le\norm{\bv G}_F^2,\]
since this corresponds to placing the center for the cloud at the origin rather than the centroid of the cloud, which incurs a higher cost. Now as a sum of squared Gaussian random variables, $\norm{\bv G}_F^2\sim\frac1{n'}\chi^2_{mn'}$, which is tightly concentrated around its mean, $m$. Therefore, $\norm{\bv A-\bv P^*\bv A}_F^2\le\norm{\bv A-\bv P\bv A}_F^2\approx m$ with high probability.

On the other hand, in $\bv A_m$, the simplex collapses to $k-1$ repeated points at the origin, so the optimal clustering corresponding to the projection $\bv{\tilde P}$ will cluster these $k-1$ points into the same cluster. We will argue that this cluster incurs a cost of almost $k=\epsilon m$.

Let us examine the first $k-1$ coordinates of this cluster centroid, corresponding to the dimensions of the simplex. Since there are at least $k-1$ points in the cluster, one of which is $\lambda'$ and the rest of which are zero, the centroid will have a coordinate of at most $\lambda'/(k-1)$ in each of these first $k-1$ coordinates. Therefore, the total squared distance of the first $k-1$ points to their cluster center will be at least
\[\lambda'^2(k-1)\left((k-2)\left(\frac1{k-1}\right)^2+\left(1-\frac1{k-1}\right)^2\right)=\frac{\lambda'^2}{k-1}(k-2+(k-2)^2)=(k-2)\lambda'^2.\]
This is about $k$, which is what we want. To technically prove the necessary claim, take $k$ large enough so that $\frac{k-2}k\ge\sqrt\lambda$, and then $n'$ large enough so that $\lambda'=1-2\sqrt{\frac k{n'\epsilon}}\ge\lambda^{1/4}$, so this cluster contributes a cost of $(k-2)\lambda'^2\ge\sqrt\lambda k\sqrt\lambda=\lambda\epsilon m$.

Now the remaining $n'$ points are clustered in $k$ clusters (possibly including the cluster used for the simplex) in $\bv{\tilde P}$ rather than one cluster as in $\bv P$ above. Then Lemma \ref{unclusterable_cloud} claims that the cloud's contribution to the cost is at least $\lambda\norm{\bv A-\bv P^*\bv A}_F^2$.

In all, therefore, we have
\[\frac{\norm{\bv A-\bv{\tilde P}\bv A}_F^2}{\norm{\bv A-\bv P^*\bv A}_F^2}\ge\lambda\epsilon+\lambda=\lambda(1+\epsilon),\]
as desired.\end{proof}

\begin{proof}[Proof of Lemma \ref{unclusterable_cloud}]One option to clustering $\bv G$ is to put just one center at the origin. We will compare this clustering to any $k$-means clustering by looking at the cost accrued to the points in each cluster. Of course, the optimal single cluster center will be the centroid of the points, but this will only perform better than the origin.
\begin{claim}For a given partition $\{\bv x_1,\dotsc,\bv x_l\}$ with $\bs \mu=\frac1l\sum_i\bv x_i$,
\[\sum_{i=1}^l\norm{\bv x_i-\bs\mu}^2=\sum_{i=1}^l\norm{\bv x_i}^2-l\norm{\bs\mu}^2.\]\end{claim}
\begin{proof}In each dimension,
\[\sum_{i=1}^l(\bv x_{i,j}-\bs\mu_j)^2=\sum_{i=1}^l\bv x_{i,j}^2-2\bs\mu_j\sum_{i=1}^l\bv x_{i,j}+l\bs\mu_j^2=\sum_{i=1}^l\bv x_{i,j}^2-2l\bs\mu_j^2+l\bs\mu_j^2=\sum_{i=1}^l\bv x_{i,j}^2-l\bs\mu_j^2.\]
Summing this over all $j$ yields the claim.
\end{proof}
Notice that the first term on the right side is the clustering cost of a single center at the origin. Therefore, the gains in the objective function from moving from clustering at the origin to $k$-means clustering with clusters $\{C_1,\dotsc,C_k \}$ are exactly $\sum_i\abs{C_i}\norm{\bs\mu(C_i)}^2$, where $\bs\mu(C_i)$ is the centroid of $C_i$. We must show that with high probability, these gains are only a $1-\lambda\ll1$ fraction of the original clustering cost.

To do so, we will argue that with high probability, no cluster will achieve large gains by concentrating in any direction. Technically, our directions will be given by a 1-net on the sphere. But first, we prove the statement for any given direction:
\begin{claim}Let $x_1,\dotsc,x_n\sim N(0,1)$ independent. Reorder the $x_i$ such that $x_1>x_2>\dotsb>x_n$. Then with high probability, $\sum_{i=1}^lx_i<10l\sqrt{\log(n/l)}$.\end{claim}
\begin{proof}The Gaussian distribution falls off superexponentially, so if $\xi\sim N(0,1)$, $\Pr(\xi>z)\le\exp(-z^2)$. Before reordering, each of the events $x_1>z,\dotsc,x_n>z$ are Bernoulli independent random variables, so their count is tightly concentrated around the mean, making $\abs{\{i:x_i>z\}}\le2n\exp(-z^2)$ with high probability.

Let $z$ take on values $z_k=2^k\sqrt{\log(n/l)}$, where $k=1,2,\dotsc,n$, so with high probability, 
\[\abs{\{i:z_{k-1}\le x_i\le z_k\}}\le2n\exp(-z_k^2)=\frac{2n}{(n/l)^{2^k}}\le\frac{2n}{(n/l)^{2k}}=\frac{2l}{(n/l)^{2k-1}}.\]
By a union bound, all of these bounds hold with high probability.

First suppose that $l\le n/2$. Then we can bound $x_1+\dotsb+x_l$ (reordered) by the contributions of the terms at most $z_k$. First, with very high probability there are no terms greater than $z_n>2^n\sqrt{\log2}$. Then for $z=1,2,\dotsc,k-1$,
\[\sum_{z_k\le x_i\le z_{k+1}}x_i\le\frac{(2l)(2^{k+1}\sqrt{\log(n/l)})}{2^{2k-1}}=\frac{8l\sqrt{\log(n/l)}}{2^k}.\]
Summing this geometric series yields a total sum of less than $8l\sqrt{\log(n/l)}$. Finally, the remaining at most $l$ terms are all less than $z_1=2\sqrt{\log(n/l)}$, so they contribute at most $2l\sqrt{\log(n/l)}$, totaling at most $10l\sqrt{\log(n/l)}$, as desired.

For $l>n/2$, it is easy to check that the bound for $n-l$ is stronger, i.e. $(n-l)\sqrt{\log(n/(n-l))}\le l\sqrt{\log(n/l)}$ for $n/2<l<n$. The total $x_1+\dotsb+x_n$ is tightly concentrated around 0 (as each $x_i$ has zero mean), so with high probability, $x_1+\dotsb+x_l\approx-x_{l+1}-\dotsb-x_n$. Since $-x_i$ has the same distribution, we can apply the result for $n-l$ to those numbers (the $n-l$ highest among the $-x_i$), and with high probability it will carry over to the desired result for $l$.\end{proof}

Now notice that if $\bv v\in S^{m-1}$ is a unit vector in some direction, its inner products with the rows of $\bv G$ are iid $N(0,1/n')$. Therefore, if $\abs {C_i}=f_in'$, i.e. if an $f_i$-fraction of the Gaussians are in the $i$th cluster, this claim shows that $\inprod{\bs\mu(C_i),\bv v}\le10\sqrt{\log(1/f_i)/n'}$ with high probability.

We now take $\bv v$ to range over a $1$-net $\mathcal N$ of $S^{m-1}$. This will have size exponential in $m$ by a simple volume argument, so we just take $n'$ to be large enough that with high probability by a union bound, $\inprod{\bs\mu(C_i),\bv v}\le10\sqrt{\log(1/f_i)/n'}$ for all $\bv v\in\mathcal N$. Now $\hat{\bs\mu}(C_i)=\bs\mu(C_i)/\norm{\bs\mu(C_i)}\in S^{m-1}$ so there exists some $\bv v\in\mathcal N$ with $\norm{\bv v-\hat{\bs\mu}(C_i)}\le1$, and expanding, $\inprod{\bv v,\hat{\bs\mu}(C_i)}\ge\frac12$. Therefore,
\[\norm{\bs\mu(C_i)}\le2\inprod{\bv v,\bs\mu(C_i)}\le20\sqrt{\log(1/f_i)/n'}.\]
Hence, with high probability, the total gain in the objective function is
\[\sum_i\abs{C_i}\norm{\bs\mu(C_i)}^2\le\sum_if_in'(400\log(1/f_i)/n')=400\sum_if_i\log(1/f_i).\]

Since $h(x) = x\log(1/x)$ is a concave function on $(0,1)$, this sum is maximized over $\sum_if_i=1$ when $f_i=1/k$ for all $i$, at which point it is equal to $400\log(k)$. Recall that the original cost is around $m=k/\epsilon$. Simply take $k$ large enough that $\frac{400\log(k)}{k/\epsilon}\le1-\lambda$, and this proves the claim.\end{proof}
\section{Approximate SVD and General Low Rank Approximation}
\label{approx_svd_proofs}
In this section we provide proofs for Theorems \ref{approx_svd} and \ref{general_svd} (stated in Section \ref{svds}). These results extend our analysis of SVD sketches from Theorem \ref{exact_svd} to the case when only an approximate SVD or general low rank approximation are available for $\bv{A}$.

\begin{reptheorem}{approx_svd}
Let $m = \lceil k/\epsilon \rceil$. For any $\bv{A} \in \mathbb{R}^{n \times d}$ and any orthonormal matrix $\bv{Z} \in \mathbb{R}^{d \times m}$ satisfying $\norm{\bv{A}-\bv{A}\bv{Z}\bv{Z}^\top}_F^2 \le (1+\epsilon') \norm{\bv{A}_{r \setminus m}}_F^2$, the sketch $\bv{\tilde A} = \bv{A}\bv{Z}\bv{Z}^\top$ satisfies the conditions of Definition \ref{def:1sidedsketch}. Specifically, for all rank $k$ orthogonal projections $\bv{P}$,
\begin{align*}
\norm{\bv{A} - \bv{P}\bv{A}}_F^2 \le \norm{\bv{\tilde A-P\tilde A}}_F^2 + c \le (1+\epsilon+\epsilon')  \norm{\bv{A} - \bv{P}\bv{A}}_F^2.
\end{align*}
\end{reptheorem}
\begin{proof}
As in the exact SVD case, since $\bv{Z}\bv{Z}^\top$ is an orthogonal projection,
\begin{align*}
\bv{\tilde C} = \bv{\tilde A} \bv{\tilde A}^\top = (\bv{A} - (\bv{A} - \bv{A}\bv{Z}\bv{Z}^\top))(\bv{A} - (\bv{A} - \bv{A}\bv{Z}\bv{Z}^\top))^\top = \bv{ A} \bv{ A}^\top - (\bv{A} - \bv{A}\bv{Z}\bv{Z}^\top)(\bv{A} - \bv{A}\bv{Z}\bv{Z}^\top)^\top.
\end{align*}
We set $\bv{E} = - (\bv{A} - \bv{A}\bv{Z}\bv{Z}^\top)(\bv{A} - \bv{A}\bv{Z}\bv{Z}^\top)^\top$. $\bv{\tilde C} = \bv{C} + \bv{E}$, $\bv{E}$ is symmetric, and $\bv{E} \preceq \bv{0}$. Finally, 
\begin{align*}
\sum_{i=1}^k |\lambda_i(\bv{E})| = \sum_{i=1}^k \sigma_i^2(\bv{A} - \bv{A}\bv{Z}\bv{Z}^\top) = \norm{(\bv{A} - \bv{A}\bv{Z}\bv{Z}^\top)_k}_F^2.
\end{align*}

Observe that, since $(\bv{A} - \bv{A}\bv{Z}\bv{Z}^\top)_k$ is rank $k$, $\bv{A}\bv{Z}\bv{Z}^\top + (\bv{A} - \bv{A}\bv{Z}\bv{Z}^\top)_k$ has rank at most $m + k$. Thus, by optimality of the SVD in low rank approximation, it must be that:
\begin{align*}
\norm{\bv{A} - \left(\bv{A}\bv{Z}\bv{Z}^\top + (\bv{A} - \bv{A}\bv{Z}\bv{Z}^\top)_k\right)}_F^2 \geq \norm{\bv{A}_{r \setminus(m+k)}}_F^2.
\end{align*}
Regrouping and applying Pythagorean theorem gives:
\begin{align*}
\norm{\bv{A} - \bv{A}\bv{Z}\bv{Z}^\top}_F^2 - \norm{(\bv{A} - \bv{A}\bv{Z}\bv{Z}^\top)_k}_F^2 \geq \norm{\bv{A}_{r \setminus (m+k)}}_F^2.
\end{align*}
Then, reordering and applying the approximate SVD requirement for $\bv{A}\bv{Z}\bv{Z}^\top$ gives
\begin{align*}
\norm{(\bv{A} - \bv{A}\bv{Z}\bv{Z}^\top)_k}_F^2 &\leq (1+\epsilon')\norm{\bv{A}_{r \setminus m}}_F^2 - \norm{\bv{A}_{r \setminus(m+k)}}_F^2 \\
&\le \epsilon'\norm{\bv{A}_{r \setminus m}}_F^2 + \sum_{i=m+1}^{m+k} \sigma_i^2(\bv{A})\\
&\le (\epsilon+\epsilon')\norm{\bv{A}_{r \setminus k}}_F^2.
\end{align*}
The last inequality follows from Equation \eqref{svd_head_bound} and the fact that 
$\norm{\bv{A}_{r \setminus k}}_F^2 \ge \norm{\bv{A}_{r \setminus m}}_F^2$. 
So, we conclude that $\sum_{i=1}^k |\lambda_i(\bv{E})| \leq (\epsilon+\epsilon')\norm{\bv{A}_{r \setminus k}}_F^2$ 
and the theorem follows from applying Lemma \ref{error_improved}.
 \end{proof}

\begin{reptheorem}{general_svd}
Let $m = \lceil k/\epsilon \rceil$. For any $\bv{A} \in \mathbb{R}^{n \times d}$ and any $\bv{\tilde A} \in \mathbb{R}^{n \times d}$ with $rank(\bv{\tilde A}) = m$ satisfying $\norm{\bv{A}-\bv{\tilde A}}_F^2 \le (1+(\epsilon')^2) \norm{\bv{A}_{r \setminus m}}_F^2$, the sketch $\bv{\tilde A}$ satisfies the conditions of Definition \ref{def:2sidedsketch}. Specifically, for all rank $k$ orthogonal projections $\bv{P}$,
\begin{align*}
 (1-2\epsilon') \norm{\bv{A} - \bv{P}\bv{A}}_F^2 \le \norm{\bv{\tilde A-P\tilde A}}_F^2 + c \le (1+2\epsilon+5\epsilon')  \norm{\bv{A} - \bv{P}\bv{A}}_F^2.
\end{align*}
\end{reptheorem}

\begin{proof}
We write $\bv{\tilde A}$ as the sum of a projection and a remainder matrix: $\bv{\tilde A} = \bv{A} \bv{Z} \bv{Z}^\top + \bv{E}$ where $\bv{Z} \in \mathbb{R}^{d \times m}$ is an orthonormal basis for row span of $\bv{\tilde A}$. 
 By the Pythagorean theorem, 
\begin{align*}
\norm{\bv{A}-\bv{\tilde A}}_F^2 &= \norm{\bv{A}-\bv{A} \bv{Z} \bv{Z}^\top}_F^2 + \norm{\bv{E}}_F^2,
\end{align*}
since the rows of $\bv{A}-\bv{A} \bv{Z} \bv{Z}^\top$ are orthogonal to the row span of $\bv{\tilde A}$ and the rows of $\bv{E}$ lie in this span. Since the SVD is optimal for low rank approximation, $\norm{\bv{A}-\bv{A} \bv{Z} \bv{Z}^\top}_F^2 \geq \norm{\bv{A}_{r \setminus m}}_F^2$. Furthermore, by our low rank approximation condition on $\bv{\tilde A}$, $\norm{\bv{A}-\bv{A} \bv{Z} \bv{Z}^\top}_F^2 \le (1+(\epsilon')^2) \norm{\bv{A}_{r \setminus m}}_F^2$. Thus:
\begin{align}
\label{r_is_small}
\norm{\bv{E}}_F^2 \le (\epsilon')^2 \norm{\bv{A}_{r \setminus m}}_F^2.
\end{align}
Also note that, by Theorem \ref{approx_svd},
\begin{align}
\label{projection_is_good}
\norm{(\bv{I-P})\bv{A}}_F^2 \le \norm{(\bv{I}- \bv{P})\bv{A} \bv{Z} \bv{Z}^\top}_F^2 +c \le (1+\epsilon + (\epsilon')^2)\norm{(\bv{I-P})\bv{A}}_F^2.
\end{align}
Using these facts, we prove Theorem \ref{general_svd}, by starting with the triangle inequality:
\begin{align*}
 \norm{(\bv{I}- \bv{P})\bv{A} \bv{Z} \bv{Z}^\top}_F - \norm{(\bv{I}- \bv{P})\bv{E}}_F \le \norm{\bv{\tilde A-P\tilde A}}_F \le \norm{(\bv{I}- \bv{P})\bv{A} \bv{Z} \bv{Z}^\top}_F + \norm{(\bv{I}- \bv{P})\bv{E}}_F.
\end{align*}
Noting that, since $\bv{I} - \bv{P}$ is a projection it can only decrease Frobenius norm, we substitute in \eqref{r_is_small}: 
\begin{align*}
\norm{\bv{\tilde A-P\tilde A}}_F ^2 &\le  \norm{(\bv{I}- \bv{P})\bv{A} \bv{Z} \bv{Z}^\top}_F^2 + \norm{\bv{E}}_F^2 + 2\norm{(\bv{I-P})\bv{A} \bv{Z} \bv{Z}^\top}_F\norm{\bv{E}}_F\\
&\le  \norm{(\bv{I}- \bv{P})\bv{A} \bv{Z} \bv{Z}^\top}_F^2 + \epsilon'^2 \norm{\bv{A}_{r \setminus m}}_F^2 + 2\epsilon' \norm{(\bv{I-P})\bv{A} \bv{Z} \bv{Z}^\top}_F\norm{\bv{A}_{r \setminus m}}_F\\
&\le (1+\epsilon')\norm{(\bv{I}- \bv{P})\bv{A} \bv{Z} \bv{Z}^\top}_F^2 + (\epsilon' + (\epsilon')^2) \norm{\bv{A}_{r \setminus m}}_F^2,
\end{align*}
where the last step follows from the AM-GM inequality. Then, using \eqref{projection_is_good} and again that $\norm{\bv{A}_{r \setminus m}}_F^2$ upper bounds $\norm{(\bv{I-P})\bv{A}}_F^2$, it follows that:
\begin{align}
\norm{\bv{\tilde A-P\tilde A}}_F ^2 &\le (1+\epsilon')(1+\epsilon+(\epsilon')^2)\norm{(\bv{I-P})\bv{A}}_F^2 - (1+\epsilon')c + (\epsilon' + (\epsilon')^2) \norm{\bv{A}_{r \setminus m}}_F^2\nonumber \\
&\le (1+\epsilon+2\epsilon' + 2(\epsilon')^2+(\epsilon')^3+\epsilon \epsilon')\norm{(\bv{I-P})\bv{A}}_F^2 - c'\nonumber \\
&\le (1+2\epsilon+5\epsilon')\norm{(\bv{I-P})\bv{A}}_F^2 -c'\label{general_upper},
\end{align}
where $c' = (1+\epsilon')c$. Our lower on $\norm{\bv{\tilde A-P\tilde A}}_F ^2$ follows similarly:
\begin{align}
\norm{\bv{\tilde A-P\tilde A}}_F ^2 &\ge  \norm{(\bv{I}- \bv{P})\bv{A} \bv{Z} \bv{Z}^\top}_F^2 - 2\norm{(\bv{I-P})\bv{A} \bv{Z} \bv{Z}^\top}_F\norm{\bv{E}}_F + \norm{\bv{E}}_F^2 \nonumber \\
&\ge  \norm{(\bv{I}- \bv{P})\bv{A} \bv{Z} \bv{Z}^\top}_F^2 - 2\epsilon'\norm{(\bv{I-P})\bv{A} \bv{Z} \bv{Z}^\top}_F\norm{\bv{A}_{r \setminus m}}_F \nonumber \\
&\ge (1-\epsilon')\norm{(\bv{I}- \bv{P})\bv{A} \bv{Z} \bv{Z}^\top}_F^2 - \epsilon'\norm{\bv{A}_{r \setminus m}}_F^2\nonumber \\
&\ge (1-\epsilon')\norm{(\bv{I-P})\bv{A}}_F^2 - (1-\epsilon') c -\epsilon'\norm{\bv{A}_{r \setminus m}}_F^2\nonumber \\
\label{general_lower}
&\ge (1-2\epsilon')\norm{(\bv{I-P})\bv{A}}_F^2 - c'.
\end{align}
The last step follows because $c' = (1+\epsilon')c \ge (1-\epsilon') c$. Combining \ref{general_upper} and \ref{general_lower} gives the result.
\end{proof}

While detailed, the analysis of Theorem \ref{general_svd} is conceptually simple -- the result relies on the small Frobenius norm of $\bv{E}$ and the triangle inequality. Alternatively, we could have computed
\begin{align*}
\bv{\tilde C} &= ( \bv{A}\bv{Z}\bv{Z}^\top+ \bv{E})(\bv{A}\bv{Z}\bv{Z}^\top+\bv{E})^\top \\
&= \bv{ A} \bv{ A}^\top - (\bv{A} - \bv{A}\bv{Z}\bv{Z}^\top)(\bv{A} - \bv{A}\bv{Z}\bv{Z}^\top)^\top + \bv{E}(\bv{A}\bv{Z}\bv{Z}^\top)^\top + (\bv{A}\bv{Z}\bv{Z}^\top)\bv{E}^\top + \bv{E}\bv{E}^\top,
\end{align*}
and analyzed it using Lemma \ref{error} directly, setting $\bv{E}_2 = -(\bv{A} - \bv{A}\bv{Z}\bv{Z}^\top)(\bv{A} - \bv{A}\bv{Z}\bv{Z}^\top)^\top + \bv{E}\bv{E}^\top$, $\bv{E}_3 = \bv{E}(\bv{A}\bv{Z}\bv{Z}^\top)^\top$, and $\bv{E}_4 = (\bv{A}\bv{Z}\bv{Z}^\top)\bv{E}^\top$. 

\section{Spectral Norm Projection-Cost Preserving Sketches}\label{spectral_version}

In this section we extend our results on sketches that preserve the Frobenius norm projection-cost, $\norm{\bv{A-PA}}_F^2$, to sketches that preserve the spectral norm cost, $\norm{\bv{A-PA}}_2^2$. Our main motivation is to prove the non-oblivious projection results of Section \ref{squish}, however spectral norm guarantees may be useful for other applications. We first give a spectral norm version of Lemma \ref{error}:

\begin{lemma}\label{error_spectral} 
For any $\bv{A} \in \mathbb{R}^{n \times d}$ and sketch $\bv{\tilde A}\in \mathbb{R}^{n \times m}$, let $\bv{C} = \bv{AA}^\top$ and $\bv{\tilde C} = \bv{\tilde A \tilde A}^\top$. If we can write $\bv{\tilde C} = \bv{C} + \bv{E}_1 + \bv{E}_2 + \bv{E}_3 + \bv{E}_4$ where
\begin{enumerate}
\item $\bv{E}_1$ is symmetric and $-\epsilon_1 \bv{C} \preceq \bv{E}_1 \preceq \epsilon_1\bv{C}$
\item $\bv{E}_2$ is symmetric, $\norm{\bv{E}_2}_2 \le \frac{\epsilon_2}{k} \norm{\bv{A}_{r \setminus k}}^2_F$
\item The columns of $\bv{E}_3$ fall in the column span of $\bv{C}$ and $\norm{\bv{E}_3^\top \bv{C}^{+} \bv{E}_3}_2 \le \frac{\epsilon_3^2}{k} \norm{\bv{A}_{r  \setminus k}}_F^2$
\item The rows of $\bv{E}_4$ fall in the row span of $\bv{C}$ and $\norm{\bv{E}_4 \bv{C}^{+} \bv{E}_4^\top}_2 \le \frac{\epsilon_4^2}{k} \norm{\bv{A}_{r \setminus k}}_F^2$
\end{enumerate}
then for any rank k orthogonal projection $\bv{P}$ and $\epsilon \geq \epsilon_1 + \epsilon_2 + \epsilon_3+ \epsilon_4$:
\begin{align*}
(1-\epsilon) \norm{\bv{ A} - \bv{P A}}_2^2 - \frac{\epsilon}{k}\norm{\bv{ A} - \bv{P A}}_F^2 \le
\norm{\bv{\tilde A} - \bv{P \tilde A}}_2^2 + c \le (1+\epsilon) \norm{\bv{ A} - \bv{P A}}_2^2 + \frac{\epsilon}{k}\norm{\bv{ A} - \bv{P A}}_F^2.
\end{align*}
\end{lemma}

\begin{proof}
Using the notation $\bv{Y} = \bv{I} - \bv{P}$ we have that $\norm{\bv{ A} - \bv{P  A}}_2^2 = \norm{\bv{ Y}\bv{C}\bv{ Y}}_2$ and $\norm{\bv{\tilde A} - \bv{P \tilde A}}_2^2 = \norm{\bv{ Y}\bv{\tilde C}\bv{ Y}}_2$. Furthermore:
\begin{align}
\label{eq:tildecbrokenout_upper_spec}
\norm{\bv{Y}\bv{\tilde C}\bv{Y}}_2 \le \norm{\bv{Y}\bv{C}\bv{Y}}_2 + \norm{\bv{Y}\bv{E}_1\bv{Y}}_2 + \norm{\bv{Y}\bv{E}_2\bv{Y}}_2 + \norm{\bv{Y}\bv{E}_3\bv{Y}}_2+ \norm{\bv{Y}\bv{E}_4\bv{Y}}_2
\end{align}
and 
\begin{align}
\label{eq:tildecbrokenout_lower_spec}
\norm{\bv{Y}\bv{\tilde C}\bv{Y}}_2 \ge \norm{\bv{Y}\bv{C}\bv{Y}}_2 - \norm{\bv{Y}\bv{E}_1\bv{Y}}_2 - \norm{\bv{Y}\bv{E}_2\bv{Y}}_2 - \norm{\bv{Y}\bv{E}_3\bv{Y}}_2 - \norm{\bv{Y}\bv{E}_4\bv{Y}}_2.
\end{align}

Our bounds on $\bv{E}_1$ immediately give $\norm{\bv{Y} \bv{E}_1 \bv{Y}}_2 \le \epsilon_1 \norm{\bv{Y} \bv{C} \bv{Y}}_2$. The spectral norm bound on $\bv{E}_2$, the fact that $\bv{Y}$ is an orthogonal projection, and the optimality of the SVD for Frobenius norm low rank approximation gives:
\begin{align*}
\norm{\bv{Y}\bv{E}_2\bv{Y}}_2 \le \norm{\bv{E}_2}_2 \le \frac{\epsilon_2}{k} \norm{\bv{A}_{r \setminus k}}^2_F \le \frac{\epsilon_2}{k} \norm{\bv{A-PA}}_F^2.
\end{align*}

Next, we note that, since $\bv{E}_3$'s columns fall in the column span of $\bv{C}$, $\bv{C}\bv{C}^+ \bv{E}_3 = \bv{E}_3$. Thus,
\begin{align*}
\norm{\bv{Y}\bv{E}_3 \bv{Y}}_2 \le \norm{\bv{Y}\bv{E}_3}_2 = \norm{(\bv{YC}) \bv{C}^+ (\bv{E}_3)}_2.
\end{align*}
We can rewrite the spectral norm as:
\begin{align*}
\norm{(\bv{YC}) \bv{C}^+ (\bv{E}_3)}_2 = \max_{\bv{a},\bv{b} \in \mathbb{R}^{n}, \norm{\bv{a}}_2 = \norm{\bv{b}}_2 = 1} \sqrt{(\bv{a}^\top \bv{YC}) \bv{C}^+ (\bv{E}_3 \bv{b})}.
\end{align*}
Since $\bv{C}^+$ is positive semidefinite, $\langle \bv{x},\bv{y} \rangle = \bv{x}^\top \bv{C}^+ \bv{y}$ is a semi-inner product and by the Cauchy-Schwarz inequality,
\begin{align*}
\norm{(\bv{YC}) \bv{C}^+ (\bv{E}_3)}_2 &\le \max_{\bv{a},\bv{b} \in \mathbb{R}^{n}, \norm{\bv{a}}_2 = \norm{\bv{b}}_2 = 1} \sqrt{(\bv{a}^\top \bv{YC}\bv{C}^+ \bv{C} \bv{Y} \bv{a})^{1/2} \cdot (\bv{b}^\top \bv{E}_3 \bv{C}^+ \bv{E}_3 \bv{b})^{1/2}}\\
& \le \sqrt{\norm{\bv{YCY}}_2 \cdot \norm{\bv{E}_3 \bv{C}^+ \bv{E}_3}_2} \\
&\le \frac{\epsilon_3}{\sqrt{k}} \norm{\bv{A-PA}}_2 \norm{\bv{A}_{r \setminus k}}_F\\
&\le \frac{\epsilon_3}{2}\norm{\bv{A-PA}}_2^2 + \frac{\epsilon_3}{2k} \norm{\bv{A-PA}}_F^2.
\end{align*}
The final inequality follows from the AM-GM inequality.
For $\bv{E}_4$ a symmetric argument gives:
\begin{align*}
\norm{\bv{Y}\bv{E}_4 \bv{Y}}_2 \le  \frac{\epsilon_4}{2}\norm{\bv{A-PA}}_2^2 + \frac{\epsilon_4}{2k} \norm{\bv{A-PA}}_F^2.
\end{align*}
Finally, combining the derived bounds for $\bv{E}_1$, $\bv{E}_2$, $\bv{E}_3$, and $\bv{E}_4$ with \eqref{eq:tildecbrokenout_upper_spec} and \eqref{eq:tildecbrokenout_lower_spec} gives:
\begin{align*}
(1-\epsilon) \norm{\bv{ A} - \bv{P A}}_2^2 - \frac{\epsilon}{k}\norm{\bv{ A} - \bv{P A}}_F^2 \le
\norm{\bv{\tilde A} - \bv{P \tilde A}}_2^2 \le (1+\epsilon) \norm{\bv{ A} - \bv{P A}}_2^2 + \frac{\epsilon}{k}\norm{\bv{ A} - \bv{P A}}_F^2.
\end{align*}
\end{proof}

It is easy to see that the conditions for Lemma \ref{error_spectral} holds for $\bv{\tilde A} = \bv{A}\bv{R}$ as long as the conditions of Lemma \ref{blocklem} are satisfied.
%
Choose $\bv{W}_1 \in \mathbb{R}^{n \times (n+m)}$ such that $\bv{W}_1\bv{B} = \bv{AZZ}^\top$ and $\bv{W}_2 \in \mathbb{R}^{n \times (n+m)}$ such that $\bv{W}_2 \bv{B} = \bv{A} - \bv{AZZ}^\top$. Recall that $\bv{E} = \bv{\tilde C} - \bv{C} = \bv{AR}\bv{R}^\top \bv{A}^\top - \bv{AA}^\top$ and thus,
\begin{align*}
\bv{E} = (\bv{W}_1 \bv{B} \bv{R} \bv{R}^\top \bv{B}^\top \bv{W}_1^\top - \bv{W}_1 \bv{B} \bv{B}^\top \bv{W}_1^\top) + 
(\bv{W}_2 \bv{B} \bv{R} \bv{R}^\top \bv{B}^\top \bv{W}_2^\top - \bv{W}_2 \bv{B} \bv{B}^\top \bv{W}_2^\top) +\\
 (\bv{W}_1 \bv{B} \bv{R} \bv{R}^\top \bv{B}^\top \bv{W}_2^\top - \bv{W}_1 \bv{B} \bv{B}^\top \bv{W}_2^\top) + 
(\bv{W}_2 \bv{B} \bv{R} \bv{R}^\top \bv{B}^\top \bv{W}_1^\top - \bv{W}_2 \bv{B} \bv{B}^\top \bv{W}_1^\top).
\end{align*}
As in Lemma \ref{blocklem}, we set $\bv{E}_1 = (\bv{W}_1 \bv{B} \bv{R} \bv{R}^\top \bv{B}^\top \bv{W}_1^\top - \bv{W}_1 \bv{B} \bv{B}^\top \bv{W}_1^\top)$ and have
\begin{align}\label{block_1_spec}
-\epsilon \bv{C} \preceq \bv{E}_1 \preceq \epsilon \bv{C}.
\end{align}
Setting $\bv{E}_2 = (\bv{W}_2 \bv{B} \bv{R} \bv{R}^\top \bv{B}^\top \bv{W}_2^\top - \bv{W}_2 \bv{B} \bv{B}^\top \bv{W}_2^\top)$, we have:
\begin{align}\label{block_2_1_spec}
\norm{\bv{E}_2}_2  = \frac{\norm{\bv{A}_{r \setminus k}}_F^2}{k} \cdot \norm{\bv{B}_2 \bv{R} \bv{R}^\top \bv{B}_2^\top - \bv{B}_2 \bv{B}_2^\top}_2 \le \frac{\epsilon}{k} \norm{\bv{A}_{r \setminus k}}_F^2.
\end{align}
Setting $\bv{E}_3 =  (\bv{W}_1 \bv{B} \bv{R} \bv{R}^\top \bv{B}^\top \bv{W}_2^\top - \bv{W}_1 \bv{B} \bv{B}^\top \bv{W}_2^\top)$, 
as shown in the proof of Lemma \ref{blocklem},
\begin{align}\label{block_3_spec}
\norm{\bv{E}_3^\top \bv{C}^+ \bv{E}_3}_2 \le \frac{\epsilon^2}{k} \norm{\bv{A}_{r \setminus k}}_F^2.
\end{align}
Finally, setting $\bv{E}_4 =  (\bv{W}_2 \bv{B} \bv{R} \bv{R}^\top \bv{B}^\top \bv{W}_1^\top - \bv{W}_2 \bv{B} \bv{B}^\top \bv{W}_1^\top) = \bv{E}_3^\top$ we have
\begin{align}\label{block_4_spec}
\norm{\bv{E}_4 \bv{C}^+ \bv{E}_4^\top}_2 \le \frac{\epsilon^2}{k} \norm{\bv{A}_{r \setminus k}}_F^2.
\end{align}
\eqref{block_1_spec}, \eqref{block_2_1_spec}, \eqref{block_3_spec}, and \eqref{block_4_spec} together ensure that $\bv{\tilde A} = \bv{AR}$ satisfies Lemma \ref{error_spectral} with error $4\epsilon$.
Together, Lemmas \ref{blocklem} and \ref{error_spectral} give a spectral norm version of Theorems \ref{rp_theorem}, \ref{sample}, and \ref{bss_theorem}:
\begin{theorem}\label{rp_theorem_spectral}
Let $\bv R \in \mathbb{R}^{d' \times d}$ be drawn from any of the matrix families of Lemma \ref{matapprox} with error $O(\epsilon)$.  Then for any matrix $\bv A \in \mathbb{R}^{n \times d}$, with probability at least $1-O(\delta)$, $\bv A \bv R^\top$ is a rank $k$ spectral norm projection-cost preserving sketch of $\bv A$ with error $\epsilon$. Specifically, for any rank $k$ orthogonal projection $\bv{P}$
\begin{align*}
(1-\epsilon) \norm{\bv{ A} - \bv{P A}}_2^2 - \frac{\epsilon}{k}\norm{\bv{ A} - \bv{P A}}_F^2 \le
\norm{\bv{\tilde A} - \bv{P \tilde A}}_2^2 \le (1+\epsilon) \norm{\bv{ A} - \bv{P A}}_2^2 + \frac{\epsilon}{k}\norm{\bv{ A} - \bv{P A}}_F^2.
\end{align*}
\end{theorem}

Applying Theorem \ref{rp_theorem_spectral} to $\bv{A}^\top$ and setting $\epsilon$ to a constant gives the requirements for Lemma \ref{spectral_norm_nonoblivious}. Note that, in general, a similar analysis to Lemma \ref{sketch_approximation} shows that a spectral norm projection-cost preserving sketch allows us to find $\bv{\tilde P}$ such that:
\begin{align*}
\norm{\bv{ A} - \bv{\tilde P A}}_2^2 \le (1+O(\epsilon)) \norm{\bv{ A} - \bv{ P^* A}}_2^2 + O\left( \frac{\epsilon}{k}\right ) \norm{\bv{ A} - \bv{ P^* A}}_F^2
\end{align*}
where $\bv{P^*}$ is the optimal projection for whatever constrained low rank approximation problem we are solving. This approximation guarantee is comparable to the guarantees achieved in \cite{Halko:2011} and \cite{bhojanapalli2014tighter} using different techniques.

\end{document}